%% file: main.tex
\documentclass[a4paper,UKenglish,cleveref, autoref, thm-restate]{lipics-v2021}

\hideLIPIcs  

\title{On the Continuity of the \texorpdfstring{\\}{} Probabilistic Bisimilarity Distance} 

\titlerunning{On the Continuity of the Probabilistic Bisimilarity Distance} 

\author{{Syyeda Zainab} Fatmi}{University of Oxford, Oxford, UK }
{}
{https://orcid.org/0000-0001-7899-8665}{Clarendon Fund}

\author{Stefan Kiefer}{University of Oxford, Oxford, UK}
{}
{https://orcid.org/0000-0003-4173-6877}{EPSRC grant EP/Z536179/1}

\author{David Parker}{University of Oxford, Oxford, UK}
{}
{https://orcid.org/0000-0003-4137-8862}{}

\author{Franck {van Breugel}}{York University, Toronto, Canada}
{}
{https://orcid.org/0009-0002-7320-1527}{Natural Sciences and Engineering Research Council of Canada}

\authorrunning{S.~Z.~Fatmi, S.~Kiefer, D.~Parker, F.~van Breugel} 

\Copyright{Syyeda Zainab Fatmi, Stefan Kiefer, David Parker, and Franck van Breugel} 



\relatedversion{This is the full version of a CONCUR 2026 paper \cite{concur}.} 

\nolinenumbers 

\usepackage[ruled,linesnumbered]{algorithm2e}
\SetKw{Let}{let}

\usepackage{booktabs}
\usepackage[table]{xcolor}
\newcolumntype{R}[1]{>{\raggedleft\let\newline\\\arraybackslash\hspace{0pt}}m{#1}}

\usepackage{pgfplots}
\usepackage{tikz}
\usetikzlibrary{fit, shapes, decorations.pathmorphing, arrows.meta, calc}
\tikzstyle{state}=[draw=black, text=black, circle, minimum size=0.8cm, inner sep=0pt]
\tikzstyle{smallstate}=[draw=black, text=black, circle, minimum size=16pt, inner sep=0pt]
\tikzstyle{bigstate}=[draw=black, text=black, circle, minimum size=1.1cm, inner sep=0pt]
\tikzstyle{diagonalstate}=[draw=black, double, text=black, circle, minimum size=1.1cm, inner sep=0pt]
\tikzstyle{class}=[rectangle, minimum width=1.7cm, minimum height=0.7cm, draw=black]

\newcommand{\vertex}{\node[circle, draw=black, inner sep=0pt, minimum size=12pt]}
\tikzset{>={Classical TikZ Rightarrow[scale=2]}, fill fraction/.style n args={2}{path picture={\fill[#1] (path picture bounding box.south west) rectangle ($(path picture bounding box.north west)!#2!(path picture bounding box.north east)$);}}}
\pgfplotsset{compat=1.18}

\DeclareMathOperator*{\argmin}{arg\,min}
\newcommand{\nat}{\mathbb{N}}

\newcommand{\support}{\mathrm{support}}
\newcommand{\Dreal}{\mathcal{D}}
\newcommand{\Sreal}{\mathcal{S}}
\newcommand{\Creal}{\Omega}
\newcommand{\Coptimal}{\Creal_{\mathrm{opt}}}
\newcommand{\proofcomment}[1]{\hspace{2em} [\mbox{#1}]}
\newcommand{\robust}{\mathord{\simeq}}
\newcommand{\mc}{\langle S, \tau \rangle}
\newcommand{\lmc}{\langle S, L, \tau, \ell \rangle}
\newcommand{\chain}[1]{\langle S \times S, #1 \rangle}
\newcommand{\reachone}[1]{\rho_1^{#1}}
\newcommand{\reachdelta}[1]{\rho_\Delta^{#1}}
\newcommand{\mvar}{V}

\newcommand{\qnext}{\textrm{next}}
\newcommand{\orange}{\textit{orange}}

\newcommand{\sem}[1]{[\![#1]\!]}

\definecolor[named]{ACMBlue}{cmyk}{1,0.1,0,0.1}
\definecolor[named]{ACMYellow}{cmyk}{0,0.16,1,0}
\definecolor[named]{ACMOrange}{cmyk}{0,0.42,1,0.01}
\definecolor[named]{ACMRed}{cmyk}{0,0.90,0.86,0}
\definecolor[named]{ACMLightBlue}{cmyk}{0.49,0.01,0,0}
\definecolor[named]{ACMGreen}{cmyk}{0.20,0,1,0.19}
\definecolor[named]{ACMPurple}{cmyk}{0.55,1,0,0.15}
\definecolor[named]{ACMDarkBlue}{cmyk}{1,0.58,0,0.21}

\def\techReport{}

\EventEditors{John Q. Open and Joan R. Access}
\EventNoEds{2}
\EventLongTitle{42nd Conference on Very Important Topics (CVIT 2016)}
\EventShortTitle{CVIT 2016}
\EventAcronym{CVIT}
\EventYear{2016}
\EventDate{December 24--27, 2016}
\EventLocation{Little Whinging, United Kingdom}
\EventLogo{}
\SeriesVolume{42}
\ArticleNo{23}

\begin{document}

\maketitle

\begin{abstract}
The probabilistic bisimilarity distance provides a quantitative measure of behavioural difference for labelled Markov chains, but it may be discontinuous under perturbations of the transition probabilities.
This lack of continuity undermines its applicability to empirically derived models, where transition probabilities are often approximations.
Recently, we introduced robust probabilistic bisimilarity as a sufficient condition for continuity at distance zero.
In this paper, we show that it is also a necessary condition, that is, two states are robustly probabilistic bisimilar if and only if their probabilistic bisimilarity distance is small for any small enough perturbation of the transition probabilities.
We further extend robustness to non-bisimilar state pairs to establish a complete characterization for continuity of the probabilistic bisimilarity distance.
Based on this characterization, we develop a polynomial time algorithm to decide continuity.
Finally, we complement our theoretical contributions with an experimental evaluation demonstrating the proposed approach in practice. Our results show that the extra step of deciding continuity requires minimal additional cost when compared to computing the probabilistic bisimilarity distance.
\end{abstract}

\input{text/introduction}
\input{text/bisimilarity}
\input{text/distance}
\input{text/robust-bisim}
\input{text/continuity}
\input{text/algorithm-continuity}
\input{text/experiments}

\input{text/conclusion}



\bibliography{main}

\ifthenelse{\isundefined{\techReport}}{
}{
\clearpage
\appendix
\input{text/appendix}
}

\end{document}

%% file: text/introduction.tex
\section{Introduction}

\emph{Probabilistic bisimilarity} (henceforth just \emph{bisimilarity}) is a concept of behavioural equivalence for probabilistic systems, such as labelled Markov chains.
Kemeny and Snell \cite{KS60} introduced the notion of lumpability for Markov chains and it was adapted to the setting of labelled Markov chains by Larsen and Skou \cite{LS89}.
Intuitively, two states of a labelled Markov chain are considered bisimilar if they have the same labels and they transition into the same equivalence classes with the same probabilities.
Bisimilarity can be used to minimize the state space of a system by identifying and merging states that are behaviourally indistinguishable, which often reduces the time required to verify properties of the system \cite{KKZJ07}.

\ifthenelse{\isundefined{\techReport}}{}{\clearpage}

However, since behavioural equivalences are sensitive to the exact values of the transition probabilities, Giacalone et al.~\cite{GJS90} proposed using behavioural pseudometrics, a quantitative generalisation of behavioural equivalences, to capture the behavioural similarity of states in a probabilistic system.
Rather than classifying states as either equivalent or inequivalent, the pseudometric maps each pair of states to a real value in the unit interval, thus also quantifying the behavioural difference between non-equivalent states.

In probabilistic verification, the most widely studied example of such a behavioural pseudometric is the \emph{probabilistic bisimilarity distance} (henceforth just \emph{bisimilarity distance}), introduced by Desharnais et al.~\cite{DGJP99}, which quantitatively generalises bisimilarity.
The lower the distance between two states, the less difference in their behaviour, and two states have distance zero if and only if they are bisimilar.

\begin{figure}[ht!]
\begin{subcaptionblock}{\textwidth}
\begin{center}
    \begin{tikzpicture}[yscale=0.9,xscale=1.2]
    \node[state] (s) at (2,3) {$s$};
    \node[state] (t) at (4,3) {$t$};
    \node[state] (t2) at (7,3) {$u$};
    \node[state, fill=ACMOrange!40] (u) at (1,1) {};
    \node[state] (v) at (3,1) {};
    \node[state] (v2) at (8,1) {};
    \node[state, fill=ACMOrange!40] (u2) at (6,1) {};
    
    \draw[->] (s) to node[above left] {$\frac{0.5}{10}\textcolor{ACMRed}{+\varepsilon}$} (u);
    \draw[->] (s) to node[above right] {$\frac{9.5}{10}\textcolor{ACMRed}{-\varepsilon}$} (v);
    \draw[->] (t) to node[below right] {\footnotesize$1$} (v);
    \draw[->] (t2) to node[above left] {$\frac{1}{10}$} (u2);
    \draw[->] (t2) to node[above right] {$\frac{9}{10}$} (v2);
    \draw[->] (u) edge[bend left=30] node[below] {\footnotesize$1$} (v);
    \draw[->] (v) edge[bend left=30] node[below] {\footnotesize$\textcolor{ACMRed}{\varepsilon}$} (u);
    \draw[->] (u2) to node[below] {\footnotesize$1$} (v2);
    \draw[->, loop below] (v2) to node[below] {\footnotesize$1$} (v2);
    \draw[->, loop below] (v) to node[below] {\footnotesize$1\textcolor{ACMRed}{-\varepsilon}$} (v);
    \end{tikzpicture}
\caption{A family of labelled Markov chains for $\varepsilon \in [0, \frac{9.5}{10}]$.  The state labels are represented by colours.}
\label{fig:intro-lmc}
\end{center}
\end{subcaptionblock}
\vspace{0.01cm}

\begin{subcaptionblock}{\textwidth}
\begin{center}
    \begin{tikzpicture}[scale=0.9, every node/.style={scale=0.9}]
    \begin{axis}
    [xlabel={$\varepsilon$}, ylabel={$\delta_{\tau_\varepsilon} (s, t)$}, width=5cm, height=4cm, at={(0.28\linewidth,0)}, ymin=-0.1,ymax=1.1]
    \addplot[domain=0:0.95,line width=1.2pt] {0.05 + x};
    \end{axis}
    \begin{axis}
    [xlabel={$\varepsilon$}, ylabel={$\delta_{\tau_\varepsilon} (s, u)$}, width=5cm, height=4cm, at={(0.72\linewidth,0)}, ymin=-0.1,ymax=1.1]
    \addplot[mark=*, mark options={scale=1}] coordinates {(0,0.05)};
    \addplot[mark=o, mark options={scale=1.2}] coordinates {(0,1)};
    \addplot[domain=0.02:0.95,line width=1.2pt] {1};
    \end{axis}
    \end{tikzpicture}
\caption{An illustration of the distances between the state pairs $(s, t)$ and $(s, u)$, respectively. When $\varepsilon = 0$, then the distance between these pairs is 0.05. As $\varepsilon$ increases, the distance between $s$ and $t$ increases proportionally, whereas the distance between $s$ and $u$ jumps up to $1$.}
\label{fig:intro-distances}
\end{center}
\end{subcaptionblock}

\caption{A motivating example.}
\label{fig:intro-example}
\end{figure}

\begin{example}
\label{example:intro}
The bisimilarity distance is based on a variation of Larsen and Skou's probabilistic modal logic~\cite{AM01,K83,LS89,MM05}.
A formula $\varphi$ of this logic assigns to each state~$s$ of a labelled Markov chain a numerical value $\sem{\varphi}(s) \in [0,1]$.
For example, consider the labelled Markov chain in Figure~\ref{fig:intro-lmc} when $\varepsilon = 0$.  The states $s$ and $t$ have a bisimilarity distance of 0.05.
Then it follows \cite[Equation~2.3]{CAMR10} that $|\sem{\varphi}(s) - \sem{\varphi}(t)| \leq 0.05$, where $\varphi$ ranges over all formulas.
Thus, the bisimilarity distance bounds the maximal difference in logical observations between states.
This is relevant for model checking.
Suppose we are interested in computing the probability of eventually reaching a state labelled with $\orange$.  This property can be expressed as the quantitative mu-calculus formula $\varphi = \mu\mvar.\,\orange \vee \qnext(\mvar)$.
For efficiency one might want to approximate these values $\sem{\varphi}$, up to some acceptable error threshold, say~$0.1$.
In this case, for states $s$ and $t$, it suffices to compute $\sem{\varphi}(s)$ only.
The same holds for \emph{every} formula $\varphi$ of the logic.
Similarly, since states $s$ and $u$ also have a bisimilarity distance of 0.05, it follows that $|\sem{\varphi}(s) - \sem{\varphi}(u)| \leq 0.05 < 0.1$, for all $\varphi$.
Hence, for states $s$ and $u$, it suffices to compute $\sem{\varphi}(s)$ only as well, speeding up model checking.\lipicsEnd
\end{example}

However, as recognised by Jaeger et al.~\cite{JMLM14}, the bisimilarity distance is sometimes not continuous.
Small changes in the transition probabilities can lead to unexpected and abrupt changes in behaviour between two states.
This is especially concerning when relying on bisimilarity as a measure of system equivalence for empirically derived models, where transition probabilities are often approximated from experimental data \cite{EH20,MLAK17,OCHTC17,SLT21,TB17}.
It may be dangerous to merge bisimilar states if the transition probabilities are not known precisely, as the distance may jump up.
In such settings, continuity analysis provides a way to assess the sensitivity of the bisimilarity distance to perturbations.
If the distance is continuous, then sufficiently small errors in the transition probabilities result in only small changes in the distance; if it is discontinuous, arbitrarily small perturbations may cause a non-negligible increase in the distance, indicating that the apparent similarity of two states may not be robust to uncertainty in the model.

More recently, in our prior work~\cite{FKPB25c}, we introduced a new notion of probabilistic equivalence, namely \emph{robust probabilistic bisimilarity} (henceforth just \emph{robust bisimilarity}).
Our main result showed that this definition ensures the continuity of the bisimilarity distance function under perturbations of the transition probabilities.
However, several fundamental questions remain unanswered.
First, although we conjectured in~\cite{FKPB25c} that robust bisimilarity is not only a sufficient but also a necessary condition for continuity, this was left open.
Second, our analysis in~\cite{FKPB25c} is restricted to bisimilar pairs of states and does not address how one could characterize and decide the continuity of the bisimilarity distance for non-bisimilar pairs of states.

In this paper, we close these gaps. We prove that robust bisimilarity exactly characterizes continuity for bisimilar state pairs, thereby settling the conjecture of~\cite{FKPB25c}.
As a consequence, we obtain a polynomial time algorithm to decide continuity for bisimilar pairs of states.
Moreover, we go beyond the setting of bisimilar states and provide a complete characterization of continuity for arbitrary pairs of states.
We also develop a polynomial-time algorithm to decide continuity for all pairs of states.
Finally, we present experimental results on an implementation of this algorithm, which demonstrate that deciding continuity incurs only a modest overhead compared to computing the bisimilarity distance, and is often significantly faster in practice.

\begin{example}
Recall the model checking scenario in \cref{example:intro} involving the states $s$, $t$, and $u$.
The argument to compute $\sem{\varphi}(s)$ only is still valid when the transition probabilities are not known precisely, \emph{provided} that the distance is continuous with respect to the transition probabilities.
As shown in \cref{fig:intro-distances}, the distance between $s$ and $t$ is continuous.  Thus, as long as the perturbation is small enough (in particular $\varepsilon < 0.05$), the distance between $s$ and $t$ is less than~$0.1$ and therefore $|\sem{\varphi}(s) - \sem{\varphi}(t)| < 0.1$.
In contrast, the distance between $s$ and $u$ is discontinuous.  Indeed, when $\varepsilon > 0$, for the formula $\varphi = \mu\mvar.\,\orange \vee \qnext(\mvar)$, we have $|\sem{\varphi}(s) - \sem{\varphi}(u)| > 0.1$.
Therefore, in this case, we cannot compute $\sem{\varphi}(s)$ only.\lipicsEnd
\end{example}

The rest of the paper is structured as follows. 
Section~\ref{section:premliminaries} provides necessary preliminaries, such as the definitions of labelled Markov chains and traditional bisimulation.
Section~\ref{section:distance} recalls the formal definition and several properties of the bisimilarity distance.
In Section~\ref{section:robust-bisimiliarity} we build on the key result from~\cite{FKPB25c} by proving the conjecture posed therein, establishing that robust bisimilarity is not only sufficient but also necessary for the continuity of the bisimilarity distance for bisimilar state pairs.
Our central technical contribution appears in Section~\ref{section:continuity}, where we complete the theoretical characterisation of the continuity of the bisimilarity distance.
In Section~\ref{section:algorithm-continuity} we provide a polynomial time algorithm to decide continuity and in Section~\ref{section:experiments} we report experimental results on the algorithm’s implementation.
Finally, Section~\ref{section:conclusion} summarises our results and discusses directions for future research.
\ifthenelse{\isundefined{\techReport}}{%
The full version of this paper, found in \cite{arxiv}, includes omitted proofs and further details.
}{%
Omitted proofs can be found in the appendices.
This paper is an extended version of \cite{concur}.
}

%% file: text/bisimilarity.tex
\section{Preliminaries}
\label{section:premliminaries}

In this section, we present some fundamental concepts that underpin this paper.
We begin with basic notions on probability distributions and Markov chains.

Let $X$ be a nonempty finite set.  A function $\mu : X \to [0, 1]$ is a \emph{subprobability distribution} on $X$ if $\sum_{x \in X} \mu(x) \leq 1$.  We denote the set of subprobability distributions on $X$ by $\Sreal(X)$.  For $\mu \in \Sreal(X)$ and $A \subseteq X$, we often write $\mu(A)$ instead of $\sum_{x \in A} \mu(x)$.  For a distribution $\mu \in \Sreal(X)$ we define the \emph{support} of $\mu$ by $\support(\mu) = \{\, x \in X \mid \mu(x) > 0 \,\}$.  A subprobability distribution $\mu$ on $X$ is a \emph{probability distribution} if $\mu(X) = 1$.  We denote the set of probability distributions on $X$ by $\Dreal(X)$.

A \emph{Markov chain} is a pair $\mc$ consisting of a finite set $S$ of states and a transition probability function $\tau : S \to \Dreal(S)$.  A \emph{labelled Markov chain} is a tuple $\lmc$ where $\mc$ is a Markov chain, $L$ is a finite set of labels and $\ell: S \to L$ is a labelling function.  A \emph{path} in a Markov chain $\mc$ is a sequence of states $s_0$, $s_1$, $s_2 \ldots$ such that $s_i \in S$ and $\tau(s_{i})(s_{i+1}) > 0$ for all $i \geq 0$.

We say that states of a Markov chain \emph{communicate} with each other if each is reachable from the other via a (possibly empty) path. This is an equivalence relation which yields a set of \emph{communication classes}.  These communication classes coincide with the strongly connected components (SCCs) of the directed graph underlying the labelled Markov chain.  A communication class is \emph{closed} if the probability of leaving the class is zero, that is, no transition with positive probability leads from a state in the class to a state outside the class.  Equivalently, a closed communication class is a bottom SCC.

For the remainder, we fix a labelled Markov chain $\lmc$, and we will study perturbations of the transition probability function $\tau$.

We now introduce couplings, which play a central role throughout the rest of this paper. Couplings provide a way to relate two probability distributions by describing how their probability mass can be ``matched.'' Intuitively, a coupling specifies a joint distribution whose marginals coincide with the original distributions.
Formally, for all $\mu$, $\nu \in \Dreal(X)$, the set $\Creal(\mu, \nu)$ of \emph{couplings} of $\mu$ and $\nu$ is defined by
\[
\Creal(\mu, \nu) = \{\, \omega \in \Dreal(X \times X) \mid \forall x \in X : \omega(x, X) = \mu(x) \mbox{ and } \omega(X, x) = \nu(x) \,\}.
\]
We write $\omega(x, X)$ for $\sum_{y \in X} \omega(x, y)$.

\begin{definition}
\label{definition:probabilistic-bisimilary}
An equivalence relation $R \subseteq S \times S$ is a \emph{bisimulation} if for all $(s, t) \in R$, $\ell(s) = \ell(t)$ and there exists $\omega \in \Creal(\tau(s), \tau(t))$ such that $\support(\omega) \subseteq R$.  States $s$ and $t$ are \emph{bisimilar}, denoted $s \sim t$, if $(s, t) \in R$ for some bisimulation~$R$.
\end{definition}

Definition~\ref{definition:probabilistic-bisimilary} \cite[Definition~4.3]{JL91} differs from the standard definition \cite[Definition~6.3]{LS89} which defines a bisimulation as an equivalence relation $R \subseteq S \times S$ such that for all $(s, t) \in R$, $\ell(s) = \ell(t)$ and for all $R$-equivalence classes $C$, $\tau(s)(C) = \tau(t)(C)$.  Nevertheless, an equivalence relation $R$ is a bisimulation by Definition~\ref{definition:probabilistic-bisimilary} if and only if it is a bisimulation as per the standard definition (see \cite[Theorem~4.6]{JL91}).
We use the coupling-based formulation of bisimulation in Definition~\ref{definition:probabilistic-bisimilary} as it aligns naturally with the definition of the bisimilarity distance introduced in the next section.

\begin{example}
Consider the states $b_1$ and $b_2$ in the labelled Markov chain of Figure~\ref{fig:lmc} with $\varepsilon = 0$. These states are bisimilar, as shown in Figure~\ref{fig:equivalence-classes}.
By the standard definition of bisimilarity, it suffices to observe that $b_1$ and $b_2$ have the same label and assign the same probability mass to each $\sim$-equivalence class: both states transition with probability $\frac{1}{2}$ to the class $\{\,c_1,c_2,c_3,c_4\,\}$ and with probability $\frac{1}{2}$ to the class $\{\,d_1,d_2\,\}$.

Equivalently, using Definition~\ref{definition:probabilistic-bisimilary}, there exists a coupling $\omega \in \Creal(\tau(b_1), \tau(b_2))$ whose support is contained in the bisimulation relation.
Concretely, one such coupling $\omega$ matches the successors of $b_1$ and $b_2$ by assigning probability $\frac{1}{2}$ each to the pairs $(c_2, c_4)$ and $(d_1, d_2)$. Since $c_2 \sim c_4$ and $d_1 \sim d_2$, it follows that $\support(\omega) \subseteq \mathord{\sim}$.\lipicsEnd
\end{example}

We write $\ell(X)$ for $\{\,\ell(x) \mid x \in X \,\}$.  If $|\ell(S)| = 1$ then $\mathord{\sim} = S \times S$.  In the remainder, we assume that the labelled Markov chain contains states with different labels, that is, $|\ell(S)| \geq 2$.  Hence, we also have that $|S| \geq 2$.

\begin{figure}[ht!]
\begin{subcaptionblock}{\textwidth}
\begin{center}
    \begin{tikzpicture}[yscale=0.9]
    \node[state] (a1) at (4,4) {$a_1$};
    \node[state,] (b1) at (9,4) {$b_1$};
    \node[state] (c1) at (2,2.5) {$c_1$};
    \node[state] (c2) at (8,2.5) {$c_2$};
    \node[state] (b2) at (6,2.5) {$b_2$};
    \node[state] (d1) at (10,2.5) {$d_1$};
    \node[state, fill=ACMBlue!30] (e1) at (3,1) {$e_1$};
    \node[state] (d2) at (5,1) {$d_2$};
    \node[state] (c4) at (7,1) {$c_4$};
    \node[state] (c3) at (1,1) {$c_3$};
    \node[state, fill=ACMBlue!30] (e2) at (9,1) {$e_2$};
    
    \draw[->] (a1) -- (c1);
    \draw[->] (b1) -- (c2);
    \draw[->] (a1) -- (b2);
    \draw[->] (b1) -- (d1);
    \draw[->] (c1) -- (c3);
    \draw[->] (c2) -- (e2);
    \draw[->] (c1) -- (e1);
    \draw[->] (c2) -- (c4);
    \draw[->] (b2) -- (d2);
    \draw[->] (b2) -- (c4);
    \draw[->] (c3) edge[bend right=25] node[below] {$\frac{1 \textcolor{ACMRed}{- \varepsilon}}{2}$} (e1);
    \draw[->] (e1) edge[bend right=25] (c3);
    \draw[->] (c4) edge[bend right=25] (e2);
    \draw[->] (e2) edge[bend right=25] (c4);
    \draw[->, ACMRed] (d1) to node[pos=0.4,below right] {$\textcolor{ACMRed}{\varepsilon}$} (e2);
    \draw[->, loop below] (c3) to node {$\frac{1 \textcolor{ACMRed}{+ \varepsilon}}{2}$} (c3);
    \draw[->, loop below] (e2) to (e2);
    \draw[->, loop below] (e1) to (e1);
    \draw[->, loop below] (d2) to node {$1$} (d2);
    \draw[->, loop below] (c4) to (c4);
    \draw[->, loop right] (d1) to node {$1 \textcolor{ACMRed}{- \varepsilon}$} (d1);
    \end{tikzpicture}
\caption{A family of labelled Markov chains for $\varepsilon \in [0, 1]$.  We denote the transition probability function for $\varepsilon = \frac{1}{n}$ by $\tau_n$ for each $n \in \nat$ and for $\varepsilon = 0$ by $\tau$.  Note that $(\tau_n)_n$ converges to $\tau$.}
\label{fig:lmc}
\end{center}
\end{subcaptionblock}
\vspace{0.01cm}

\begin{subcaptionblock}{\textwidth}
\begin{center}
    \begin{tikzpicture}[yscale=0.9]
    \node[class, minimum width=1cm] (a) at (3.35,0) {$a_1$};
    \node[class] (b) at (1.4,0) {$b_1$, $b_2$};
    \node[class, minimum width=3cm] (c) at (6,0) {$c_1$, $c_2$, $c_3$, $c_4$};
    \node[class] (d) at (-0.9,0) {$d_1$, $d_2$};
    \node[class, fill=ACMBlue!30] (e) at (9,0) {$e_1$, $e_2$};

    \draw[->] (a) -- (b);
    \draw[->] (a) -- (c);
    \draw[-{Classical TikZ Rightarrow[scale=2]}] (b.south east) to[bend right=25] (c.south west);
    \draw[->] (b) -- (d);
    \draw[-{Classical TikZ Rightarrow[scale=2]}] ($(e.north west)!0.3!(e.south west)$) to ($(c.north east)!0.3!(c.south east)$);
    \draw[-{Classical TikZ Rightarrow[scale=2]}] ($(c.north east)!0.7!(c.south east)$) to ($(e.north west)!0.7!(e.south west)$);
    \draw[->, loop below] (d) to node[left, xshift=-3] {$1$} (d);
    \draw[->, loop below] (c) to (c);
    \draw[->, loop below] (e) to (e);
    \end{tikzpicture}
\caption{The equivalence classes under traditional bisimilarity for $\tau$.  For each equivalence class, the outgoing transitions shown are those induced by the transitions of its constituent states.}
\label{fig:equivalence-classes}
\end{center}
\end{subcaptionblock}
\vspace{0.01cm}

\begin{subcaptionblock}{\textwidth}
\begin{center}
    \begin{tikzpicture}[yscale=0.8] 
    \node[bigstate] (1) at (5,9.5) {$a_1 \; b_2$};
    \node[bigstate] (2) at (3,9.5) {$c_1 \; c_4$};
    \node[bigstate] (3) at (7,9.5) {$b_2 \; d_2$};
    \node[bigstate] (4) at (2,7.25) {$c_3 \; c_4$};
    \node[bigstate, fill=ACMBlue!30] (5) at (4,7.25) {$e_1 \; e_2$};
    \node[diagonalstate] (6) at (6,7.25) {$d_2 \; d_2$};
    \node[bigstate] (7) at (8,7.25) {$c_4 \; d_2$};
    \node[bigstate, fill fraction={ACMBlue!30}{0.5}] (8) at (10,7.25) {$e_2 \; d_2$};
    \node[bigstate, fill=ACMBlue!30, fill fraction={white}{0.5}] (0) at (0,7.25) {$c_3 \; e_2$};
    \draw (0.south) -- (0.north) ;
    \draw (1.south) -- (1.north) ;
    \draw (2.south) -- (2.north) ;
    \draw (3.south) -- (3.north) ;
    \draw (4.south) -- (4.north) ;
    \draw (5.south) -- (5.north) ;
    \draw (6.south) -- (6.north) ;
    \draw (7.south) -- (7.north) ;
    \draw (8.south) -- (8.north) ;
    \draw[->] (1) edge (2);
    \draw[->] (1) edge (3);
    \draw[->] (2) edge (4);
    \draw[->] (2) edge (5);
    \draw[->] (3) edge (6);
    \draw[->] (3) edge (7);
    \draw[->, ACMRed] (4) edge node[below] {$\frac{\textcolor{ACMRed}{\varepsilon}}{2}$} (0);
    \draw[->] (4) edge[bend right=25] node[below] {$\frac{1 \textcolor{ACMRed}{- \varepsilon}}{2}$} (5);
    \draw[->] (5) edge[bend right=25] (4);
    \draw[->, loop below] (4) edge (4);
    \draw[->, loop below] (5) edge (5);
    \draw[->, loop below] (6) edge node[below] {$1$} (6);
    \draw[->, loop below] (7) edge (7);
    \draw[->] (7) edge (8);
    \draw[->, loop below] (8) edge node[below] {$1$} (8);
    \draw[->, loop below] (0) edge node[below] {$1$} (0);
    
    \node[bigstate] (2) at (8.5,4.5) {$b_1 \; b_2$};
    \node[bigstate] (4) at (7,2.75) {$c_2 \; c_4$};
    \node[bigstate] (5) at (10,2.75) {$d_1 \; d_2$};
    \node[diagonalstate] (6) at (6,0.5) {$c_4 \; c_4$};
    \node[diagonalstate, fill=ACMBlue!30] (7) at (8,0.5) {$e_2 \; e_2$};
    \node[bigstate, fill fraction={ACMBlue!30}{0.5}] (8) at (10,0.5) {$e_1 \; d_2$};
    \draw (2.south) -- (2.north) ;
    \draw (4.south) -- (4.north) ;
    \draw (5.south) -- (5.north) ;
    \draw (6.south) -- (6.north) ;
    \draw (7.south) -- (7.north) ;
    \draw (8.south) -- (8.north) ;
    \draw[->] (2) edge (4);
    \draw[->] (2) edge (5);
    \draw[->] (4) edge (6);
    \draw[->] (4) edge (7);
    \draw[->, ACMRed] (5) edge node[right] {$\textcolor{ACMRed}{\varepsilon}$} (8);
    \draw[->] (6) edge[bend right=25] (7);
    \draw[->] (7) edge[bend right=25] (6);
    \draw[->, loop above] (5) edge node[above] {$1 \textcolor{ACMRed}{- \varepsilon}$} (5);
    \draw[->, loop below] (6) edge (6);
    \draw[->, loop below] (7) edge (7);
    \draw[->, loop below] (8) edge node[below] {$1$} (8);

    \node[bigstate] (1) at (3.25,4.5) {$a_1 \; c_1$};
    \node[bigstate, fill=ACMBlue!30, fill fraction={white}{0.5}] (3) at (4.5,2.75) {$b_2 \; e_1$};
    \node[bigstate, fill fraction={ACMBlue!30}{0.5}] (9) at (0,2.75) {$e_1 \; c_3$};
    \node[bigstate] (10) at (2,2.75) {$c_1 \; c_3$};
    \node[diagonalstate] (11) at (1,0.5) {$c_3 \; c_3$};
    \node[diagonalstate, fill=ACMBlue!30] (12) at (3,0.5) {$e_1 \; e_1$};
    \draw (1.south) -- (1.north) ;
    \draw (3.south) -- (3.north) ;
    \draw (9.south) -- (9.north) ;
    \draw (10.south) -- (10.north) ;
    \draw (11.south) -- (11.north) ;
    \draw (12.south) -- (12.north) ;
    \draw[->, loop below] (3) edge node[below] {$1$} (3);
    \draw[->] (1) edge (3);
    \draw[->] (1) edge (10);
    \draw[->] (10) edge (11);
    \draw[->] (10) edge node[pos=0.6, above right] {$\frac{1 \textcolor{ACMRed}{- \varepsilon}}{2}$} (12);
    \draw[->, ACMRed] (10) edge node[above] {$\textcolor{ACMRed}{\frac{\varepsilon}{2}}$} (9);
    \draw[->] (11) edge[bend right=25] node[below] {$\frac{1 \textcolor{ACMRed}{- \varepsilon}}{2}$} (12);
    \draw[->] (12) edge[bend right=25] (11);
    \draw[->, loop below] (11) edge node[below] {$\frac{1 \textcolor{ACMRed}{+ \varepsilon}}{2}$} (11);
    \draw[->, loop below] (12) edge (12);
    \draw[->, loop above] (9) edge node[above] {$1$} (9);
    \end{tikzpicture}
\caption{Part of the Markov chain $\chain{P_\varepsilon}$ induced by a policy $P_\varepsilon$ that is optimal for all $\varepsilon \in [0, 1]$.}
\label{fig:policy}
\end{center}
\end{subcaptionblock}

\caption{The running example. All omitted transition probabilities are $\frac{1}{2}$.}
\label{fig:running-example}
\end{figure}

%% file: text/distance.tex
\section{Probabilistic Bisimilarity Distance}
\label{section:distance}

The bisimilarity distance \cite{DGJP99} measures the behavioural difference between two states. A smaller distance indicates more similar behaviour, with a distance of zero corresponding precisely to behavioural equivalence.

Before presenting the formal definition of the bisimilarity distance, which is phrased as a fixed point, we first outline the underlying idea.
If a pair of states, say $s$ and $t$, have different labels, the states are maximally different and, therefore, have distance $1$.
Otherwise, one compares their transition distributions. This is done by selecting a coupling that matches successors of $s$ with successors of $t$ in an optimal way, and then measuring the behavioural difference of these successor pairs.
As we will see in Proposition~\ref{proposition:reach-s1} below, the distance between two states corresponds to the minimal probability that, when exploring the two states jointly, one eventually encounters a pair of states with different labels.

\begin{definition}
\label{definition:delta}
The \emph{bisimilarity distance}, $\delta_\tau : S \times S \to [0, 1]$, is the least fixed point of the function $\Delta_\tau : (S \times S \to [0, 1]) \to (S \times S \to [0, 1])$ defined by
\[
\Delta_\tau(d)(s, t) = \left \{
\begin{array}{ll}
1 & \hspace{0.5cm} \mbox{if $\ell(s) \neq \ell(t)$}\\
\displaystyle \inf_{\omega \in \Creal(\tau(s), \tau(t))} \sum_{u, v \in S} \omega(u, v) \; d(u, v) & \hspace{0.5cm} \mbox{otherwise.}
\end{array}
\right .
\]
\end{definition}

\begin{theorem}[{\cite[Theorem 2.1.30]{T18}}]
\label{theorem:pseudometric}
$\delta_\tau$ is a pseudometric.
\end{theorem}

\begin{theorem}[{\cite[Theorem~4.10]{DGJP04}}]
\label{theorem:distance-zero}
For all $s$, $t \in S$, $s \sim t$ if and only if $\delta_\tau(s, t) = 0$.
\end{theorem}

The bisimilarity distance function $\delta_{\_}(s, t)$ is \emph{lower semi-continuous} at $\tau$ if for any sequence $(\tau_n)_n$ converging to $\tau$ we have $\liminf_n \delta_{\tau_n}(s, t) \geq \delta_\tau(s, t)$ and \emph{upper semi-continuous} at $\tau$ if we have $\limsup_n \delta_{\tau_n}(s, t) \leq \delta_\tau(s, t)$.  Lastly, $\delta_{\_}(s, t)$ is \emph{continuous} at $\tau$ if it is both lower semi-continuous and upper semi-continuous at $\tau$. The function $\delta_{\_}(s, t)$ is lower semi-continuous (upper semi-continuous, continuous, resp.) if it is lower semi-continuous (upper semi-continuous, continuous, resp.) at $\tau$ for all transition probability functions $\tau$.

\begin{proposition}[{\cite[Proposition~1]{FKPB25c}}]
\label{proposition:lower-semi-continuous}
For all $s$, $t \in S$, the function $\delta_{\_}(s, t) : (S \to \Dreal(S)) \to [0, 1]$ is lower semi-continuous.
\end{proposition}

The following subsets of $S \times S$ play a key role in the subsequent discussion.
\begin{definition}
\label{definition:sets}
The sets $S^2_\Delta$, $S^2_{0,\tau}$, $S^2_1$, $S^2_{?,\tau}$, and $S^2_{0?}$ are defined by
\begin{align*}
S^2_\Delta = &~ \{\, (s, s) \mid s \in S \,\}\\
S^2_{0,\tau} = &~ \{\, (s, t) \in S \times S \mid s \not= t \wedge s \sim t \,\}\\
S^2_1 = &~ \{\, (s, t) \in S \times S \mid \ell(s) \not= \ell(t) \,\}\\
S^2_{?,\tau} = &~ (S \times S) \setminus (S^2_\Delta \cup S^2_{0,\tau} \cup S^2_1)\\
S^2_{0?} = &~ S^2_{0,\tau} \cup S^2_{?,\tau}
\end{align*}
\end{definition}
The first four sets form a partition of $S \times S$.  Observe that the sets $S^2_{0,\tau}$ and $S^2_{?,\tau}$ depend on $\tau$ and may, therefore, change when we perturb $\tau$, whereas the sets $S^2_\Delta$ and $S^2_1$ stay the same.  Note that $S^2_{0?} = (S \times S) \setminus (S^2_\Delta \cup S^2_1)$.  Hence, this set also stays the same if we perturb $\tau$.  Furthermore, note that $\mathord{\sim} = S^2_\Delta \cup S^2_{0,\tau}$ and for all $(s, t) \in S^2_1$, we have $\delta_\tau(s,t) = 1$.

\begin{definition}
Let $\tau : S \to \Dreal(S)$.  The set $\mathcal{P}_\tau$ of \emph{policies} for $\tau$ is defined by
\[
\mathcal{P}_\tau
=
\left \{\, P : (S \times S) \to \Dreal(S \times S) ~ \middle \vert
\begin{array}{l}
\forall (s, t) \in S^2_\Delta \cup S^2_{0?} : P(s, t) \in \Creal(\tau(s), \tau(t))\\
\forall (s, t) \in S^2_1 : \support(P(s, t)) = \{ (s, t) \}
\end{array}
\right \}.
\]
\end{definition}
Note that a policy $P \in \mathcal{P}_\tau$ induces a Markov chain $\chain{P}$.  The subscript $\tau$ is omitted when clear from the context.

Intuitively, a policy chooses, for each pair of states $(s, t) \in S^2_\Delta \cup S^2_{0?}$, a joint transition distribution, which allows us to reason about reachability probabilities in the induced Markov chain on the product state space $S \times S$.

The following proposition characterizes $\delta_\tau$ in terms of policies.
\begin{proposition}
\label{proposition:reach-s1}
For all $s$, $t \in S$, $\displaystyle \delta_\tau(s, t) = \min_{P \in \mathcal{P}} \reachone{P}(s, t)$, where $\reachone{P}(s, t)$ is the probability with which $(s, t)$ reaches $S_1^2$ in $\chain{P}$.
\end{proposition}
The proof of Proposition~\ref{proposition:reach-s1} follows from {\cite[Theorem~8]{CBW12}} and {\cite[Theorem~10.15]{BK08}}.
This characterization reduces the distance to a reachability problem in a suitably constructed Markov chain.

We say that a policy $P \in \mathcal{P}$ is \emph{optimal} if for all $s$, $t \in S$, we have $\reachone{P}(s, t) = \delta_\tau(s, t)$. For $s$, $t \in S$, we say that $\omega \in \Creal(\tau(s), \tau(t))$ is \emph{optimal} if $\delta_\tau(s, t) = \omega \cdot \delta_\tau \overset{\text{def}}{=} \sum_{u, v \in S} \omega(u, v) \; \delta_\tau(u, v)$.  For $s$, $t \in S$, we denote the set of optimal couplings of $\tau(s)$ and $\tau(t)$ by $\Coptimal(\tau(s), \tau(t))$.  Note that $\Coptimal(\tau(s), \tau(t))$ is nonempty (see 
\ifthenelse{\isundefined{\techReport}}{\cite[Remark~42]{arxiv}}{\cref{remark:optimal-couplings}}).

%% file: text/robust-bisim.tex
\section{Robust Probabilistic Bisimilarity}
\label{section:robust-bisimiliarity}

Robust bisimilarity \cite{FKPB25c} is an equivalence relation that is a subset of bisimilarity, as it also accounts for uncertainty in the transition probabilities.
Two states that are robustly bisimilar remain almost behaviourally equivalent under small perturbations of the transition probabilities, that is, their bisimilarity distance remains small.

\begin{definition}[{\cite[Definition~5]{FKPB25c}}]
\label{definition:robust-probabilistic-bisimilarity}
\emph{Robust bisimilarity}, denoted $\robust$, is defined for $s$, $t \in S$ as $s \simeq t$ if there exists a policy $P \in \mathcal{P}$ such that
$\reachdelta{P}(s, t) = 1$, where $\reachdelta{P}(s, t)$ is the probability of reaching $S^2_\Delta$ from $(s, t)$ in $\chain{P}$.
\end{definition}

The following theorem states that robust bisimilarity is a sufficient condition for the continuity of the bisimilarity distance function for bisimilar pairs of states.

\begin{theorem}[{\cite[Theorem~2]{FKPB25c}}]
\label{theorem:cav-theorem}
For all $s$, $t \in S$, if $s \simeq t$ then $s \sim t$, and $\delta_{\_}(s, t) : (S \to \Dreal(S)) \to [0, 1]$ is continuous at $\tau$.
\end{theorem}

The following theorem provides a converse, closing a gap identified in~\cite{FKPB25c}.

\begin{theorem}
\label{theorem:cav-conjecture}
For all $s$, $t \in S$, if $s \sim t$ then $s \simeq t$ holds if and only if $\delta_{\_}(s, t) : (S \to \Dreal(S)) \to [0, 1]$ is continuous at $\tau$.
\end{theorem}
We prove Theorem~\ref{theorem:cav-conjecture} in Section~\ref{section:continuity}.

\begin{example}
\label{example:robust-bisimilarity}
Consider the family of labelled Markov chains in Figure~\ref{fig:lmc}.  Let $\varepsilon = 0$.  The equivalence classes under traditional bisimilarity for $\tau$ are displayed in Figure~\ref{fig:equivalence-classes}.
Among these, the only non-trivial robustly bisimilar sets of states are $\{ c_1, c_3 \}$ and $\{ c_2, c_4 \}$.  Observe that each of these pairs reach $S^2_\Delta$ with probability $1$ in $\chain{P_0}$, where the relevant part of $P_0 \in \mathcal{P}_\tau$ is illustrated in Figure~\ref{fig:policy}.

We analyse the continuity of the bisimilarity distance function at $\tau$ for three pairs of bisimilar states;
see \cref{fig:robust-bisimilarity-example} for an illustration.
For the robustly bisimilar pair $(c_1, c_3)$, we have $\delta_{\tau_n}(c_1, c_3) = \frac{\varepsilon}{2} = \frac{1}{2n}$ for all $n \in \nat$. Note that if $n$ is large then the distance is small and $\lim_n \delta_{\tau_n}(c_1, c_3) = 0 = \delta_\tau(c_1, c_3)$.  Thus, $\delta_{\_}(c_1, c_3) : (S \to \Dreal(S)) \to [0, 1]$ is continuous at $\tau$.

On the other hand, the pairs $(b_1, b_2)$ and $(c_1, c_4)$ are bisimilar but not robustly so. 
Indeed, the pair $(b_1,b_2)$ reaches $S^2_\Delta$ with probability only $\frac{1}{2}$ in $\chain{P_0}$, and no policy can improve this.
Due to the $\varepsilon$ transition from state $d_1$ to $e_2$, we have $\delta_{\tau_n}(d_1, d_2) = 1$ for all $n \in \nat$.
It follows that $\delta_{\tau_n}(b_1, b_2) = \frac{1}{2}$ for all $n \in \nat$.
Consequently, the sequence $\delta_{\tau_n}(b_1,b_2)$ does not converge to $\delta_{\tau}(b_1,b_2) = 0$.

Similarly, the pair $(c_1, c_4)$ cannot reach $S^2_\Delta$ under any policy.  Due to the perturbation of the outgoing transitions from $c_3$, we have $\delta_{\tau_n}(c_3, c_4) = 1$ and $\delta_{\tau_n}(e_1, e_2) = 1$ for all $n \in \nat$.
Hence, $\delta_{\tau_n}(c_1, c_4) = 1$ for all $n \in \nat$.
Therefore, $\delta_{\_}(b_1, b_2) : (S \to \Dreal(S)) \to [0, 1]$ and $\delta_{\_}(c_1, c_4) : (S \to \Dreal(S)) \to [0, 1]$ are discontinuous at $\tau$.\lipicsEnd
\end{example}

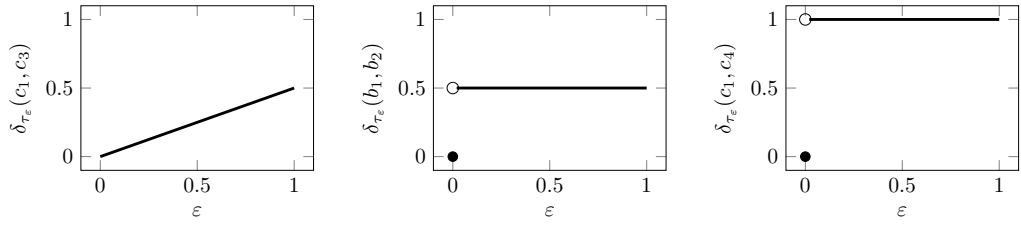
\begin{figure}[ht!]
\begin{center}
    \begin{tikzpicture}[scale=0.9, every node/.style={scale=0.9}]
    \begin{axis}
    [xlabel={$\varepsilon$}, ylabel={$\delta_{\tau_\varepsilon} (c_1, c_3)$}, width=5cm, height=4cm, at={(0.13\linewidth,0)}, ymin=-0.1,ymax=1.1]
    \addplot[domain=0:1,line width=1.2pt] {x/2};
    \end{axis}
    \begin{axis}
    [xlabel={$\varepsilon$}, ylabel={$\delta_{\tau_\varepsilon} (b_1, b_2)$}, width=5cm, height=4cm, at={(0.5\linewidth,0)}, ymin=-0.1,ymax=1.1]
    \addplot[mark=*, mark options={scale=1}] coordinates {(0,0)};
    \addplot[mark=o, mark options={scale=1.2}] coordinates {(0,0.5)};
    \addplot[domain=0.02:1,line width=1.2pt] {0.5};
    \end{axis}
    \begin{axis}
    [xlabel={$\varepsilon$}, ylabel={$\delta_{\tau_\varepsilon} (c_1, c_4)$}, width=5cm, height=4cm, at={(0.87\linewidth,0)}, ymin=-0.1,ymax=1.1]
    \addplot[mark=*, mark options={scale=1}] coordinates {(0,0)};
    \addplot[mark=o, mark options={scale=1.2}] coordinates {(0,1)};
    \addplot[domain=0.02:1,line width=1.2pt] {1};
    \end{axis}
    \end{tikzpicture}
\end{center}
\caption{An illustration of the distances between the state pairs $(c_1, c_3)$, $(b_1, b_2)$, and $(c_1, c_4)$ from the labelled Markov chain in \cref{fig:running-example}.  All three state pairs are bisimilar, although only $(c_1, c_3)$ is robustly bisimilar.  It follows that the distance is continuous at $\tau$ for $(c_1, c_3)$ only.}
\label{fig:robust-bisimilarity-example}
\end{figure}

\begin{theorem}
\label{theorem:bisim-continuity-poly-time}
Continuity for bisimilar state pairs can be decided in polynomial time.
\end{theorem}
\begin{proof}
Robust bisimilarity can be computed in polynomial time \cite{FKPB25c}. Then the result follows from \cref{theorem:cav-conjecture}.
\end{proof}
 

%% file: text/continuity.tex
\section{Continuity}
\label{section:continuity}

In this section, we extend the characterization of the continuity of the bisimilarity distance function to non-bisimilar state pairs.

\begin{theorem}
\label{theorem:main-continuity}
Let $s$, $t \in S$.
\begin{enumerate}[a.]
    \item For all $P \in \mathcal{P}$, we have $\reachdelta{P}(s, t) \leq 1 - \delta_\tau (s, t)$.
    \item The function $\delta_{\_}(s, t) : (S \to \Dreal(S)) \to [0, 1]$ is continuous at $\tau$ if and only if there exists a policy $P \in \mathcal{P}$ such that $\reachdelta{P}(s, t) = 1 - \delta_{\tau}(s, t)$.
\end{enumerate}
\end{theorem}
\begin{proof}[Proof Sketch]
Let $s$, $t \in S$.
\begin{enumerate}[a.]
\item Towards a contradiction, assume that there exists a policy $P \in \mathcal{P}$ such that $\reachdelta{P}(s, t) > 1 - \delta_\tau (s, t)$.  One could then modify such a policy so that pairs of states in $S^2_\Delta$ never leave $S^2_\Delta$. This would produce a policy $P' \in \mathcal{P}$ where $\reachone{P'}(s, t) \leq 1 - \reachdelta{P}(s, t) < \delta_\tau (s, t)$, contradicting Proposition~\ref{proposition:reach-s1}.
\item
``$\Leftarrow$'': The proof of this direction is similar to the one of Theorem~\ref{theorem:cav-theorem}.  Let us give a sketch here.

Let $(\tau_n)_n$ be a sequence converging to $\tau$.  Assume that $P \in \mathcal{P}_\tau$ is a policy such that $\reachdelta{P}(s, t) = 1 - \delta_\tau(s, t)$.  Then $P$ is optimal for $(s, t)$, that is, $\reachone{P} (s, t) = \delta_\tau(s, t)$.  We construct a graph consisting of the communication classes of $\chain{P}$ that are reachable from $(s, t)$.  Let $P_n \in \mathcal{P}_{\tau_n}$.  We then show that for all communication classes $C$ reachable from $(s, t)$ and for all pairs $(u, v) \in C$, it holds that $\lim_n \reachone{P_n}(u, v) = \reachone{P}(u, v)$, by induction on the length of a longest path from $C$ in the above-mentioned graph.

By Proposition~\ref{proposition:reach-s1}, $\limsup_n \delta_{\tau_n}(s, t) \leq \limsup_n \reachone{P_n}(s, t)$.  Using the above results, we conclude that $\limsup_n \delta_{\tau_n}(s, t) \leq \reachone{P} (s, t) = \delta_\tau(s, t)$.

``$\Rightarrow$'': Assume that for all policies $P \in \mathcal{P}$, we have $\reachdelta{P}(s, t) < 1 - \delta_\tau (s, t)$.

We construct a sequence $(\tau_n)_n$ converging to $\tau$ as follows.
Assign each state a unique positive integer \emph{index}.
Let $n \in \nat$.
For each state, move \emph{index}$\cdot \varepsilon_n$ probability from an existing successor to a state with a different label, where $\varepsilon_n > 0$ is a small number.
This forces each pair of bisimilar states to transition to $S^2_1$ with positive probability, thereby destroying $\sim$.
See Figure~\ref{fig:perturbation} for an example of such a perturbation.

Since the state space is finite, there exists a subsequence $(P_{f(n)})_n$ of optimal policies with $P_{f(n)} \in \mathcal{P}_{\tau_{f(n)}}$ that all share the same support graph.
Then the integer indices ensure that there exist a limit policy $P \in \mathcal{P}_\tau$ and a function $D : (S \times S) \to (S \times S) \to \mathbb{Z}$ such that $P_{f(n)} = P + \varepsilon_n \, D$.
A visualisation of these policies appears in Figure~\ref{fig:perturbation-policy}.
Since $S \times S \to [0, 1]$ is compact, the sequence $\left( \reachone{P_{f(n)}} \right)_n$ has a converging subsequence $\left( \reachone{P_{g(n)}} \right)_n$.

By \cref{proposition:reach-s1}, it suffices to consider the case where $P$ is optimal.
By the assumption above (next to ``$\Rightarrow$''), we have $\reachdelta{P}(s, t) < 1 - \delta_\tau (s, t)$, which then implies that $(s, t)$ can reach a closed communication class $C$ of $\chain{P}$ with $C \subseteq S^2_{0,\tau}$.  Due to the nature of the perturbation, for all $(u, v) \in C$ we have $\reachone{P}(u, v) = 0 < \liminf_n \reachone{P_{g(n)}}(u, v)$.  Consider a path from $(s, t)$ to $C$.  For each $(x, y)$ on this path, we prove that $\reachone{P}(x, y) < \lim_n \reachone{P_{g(n)}}(x, y)$, by induction on the length of the path from $(x, y)$ to $C$.

Since every $P_{g(n)}$ is optimal, from the above we can conclude that $\delta_\tau(s, t) \leq \reachone{P}(s, t) < \lim_n \reachone{P_{g(n)}}(s, t) = \lim_n \delta_{\tau_{\varepsilon_{g(n)}}}(s, t)$.
Hence, $\delta_{\_}(s, t) : (S \to \Dreal(S)) \to [0, 1]$ is discontinuous at $\tau$.

As an example, consider the bisimilar pair $(u, v)$ in Figure~\ref{fig:perturbation-policy}, which forms a closed communication class of $\chain{P}$. We have $\reachone{P}(u, v) = 0$ and $\liminf_n \reachone{P_{g(n)}}(u, v) = \frac{1}{4}$.  It follows that $\reachone{P}(s, t) = \frac{1}{4} < \lim_n \reachone{P_{g(n)}}(s, t) = \frac{3}{8}$. \qedhere
\end{enumerate}
\end{proof}

\begin{figure}[ht]
\begin{subcaptionblock}{\textwidth}
\begin{center}
    \begin{tikzpicture}[xscale=1.2]
    \node[state] (s) at (2,2.75) {$s$};
    \node[state] (t) at (4,2.75) {$t$};
    \node[state, fill=ACMPurple!25] (w) at (3,1) {$w$};
    \node[state] (v) at (5,1) {$v$};
    \node[state] (u) at (1,1) {$u$};

    \draw[->] (s) to node[above left, xshift=3] {$\frac{1}{2} \textcolor{ACMRed}{+ \varepsilon}$} (u);
    \draw[->] (t) to node[above right, xshift=-3] {$\frac{3}{4} \textcolor{ACMRed}{+ 2\varepsilon}$} (v);
    \draw[->] (s) to node[pos=0.3, below left, xshift=3] {$\frac{1}{2} \textcolor{ACMRed}{- \varepsilon}$} (w);
    \draw[->] (t) to node[above left, xshift=3] {$\frac{1}{4} \textcolor{ACMRed}{- 2\varepsilon}$} (w);
    \draw[->, ACMRed] (u) edge[bend right=25] node[below] {$3\varepsilon$} (w);
    \draw[->, ACMRed] (v) edge[bend left=25] node[below] {$4\varepsilon$} (w);
    \draw[->, ACMRed] (w) edge[bend left=25] node[below] {$5\varepsilon$} (v);
    \draw[->, loop below] (u) to node {$1 \textcolor{ACMRed}{- 3\varepsilon}$} (u);
    \draw[->, loop below] (v) to node {$1 \textcolor{ACMRed}{- 4\varepsilon}$} (v);
    \draw[->, loop below] (w) to node {$1 \textcolor{ACMRed}{- 5\varepsilon}$} (w);
    \end{tikzpicture}
\caption{A family of labelled Markov chains for $\varepsilon \in [0, \frac{1}{8}]$.  We denote the transition probability function for $\varepsilon = 0$ by $\tau$.}
\label{fig:perturbation}
\end{center}
\end{subcaptionblock}
\vspace{0.01cm}

\begin{subcaptionblock}{\textwidth}
\begin{center}
    \begin{tikzpicture}[yscale=0.8,xscale=1.1]
    \node[bigstate] (2) at (3,4.5) {$s \quad t$};
    \node[bigstate, fill fraction={ACMPurple!25}{0.5}] (1) at (0.75,4.5) {$w \;\;\; v$};
    \node[bigstate, fill=ACMPurple!25, fill fraction={white}{0.5}] (3) at (0,2) {$u \;\;\; w$};
    \node[bigstate] (4) at (2,2) {$u \quad v$};
    \node[diagonalstate, fill=ACMPurple!25] (5) at (4,2) {$w \;\; w$};
    \node[diagonalstate] (6) at (6,2) {$v \quad v$};
    \draw (1.south) -- (1.north) ;
    \draw (2.south) -- (2.north) ;
    \draw (3.south) -- (3.north) ;
    \draw (4.south) -- (4.north) ;
    \draw (5.south) -- (5.north) ;
    \draw (6.south) -- (6.north) ;
    \draw[->] (2) edge node[above] {$\frac{1}{4} \textcolor{ACMRed}{+ \varepsilon}$} (1);
    \draw[->] (2) edge node[pos=0.6, above left] {$\frac{1}{2} \textcolor{ACMRed}{+ \varepsilon}$} (4);
    \draw[->] (2) edge node[pos=0.6, above right] {$\frac{1}{4} \textcolor{ACMRed}{- 2\varepsilon}$} (5);
    \draw[->, ACMRed] (4) edge node[below] {$3\varepsilon$} (5);
    \draw[->, ACMRed] (5) edge[bend left=25] node[above] {$5\varepsilon$} (6);
    \draw[->, ACMRed] (6) edge[bend left=25] node[below] {$4\varepsilon$} (5);
    \draw[->, loop below] (4) edge node[below] {$1 \textcolor{ACMRed}{- 4\varepsilon}$} (4);
    \draw[->, loop below] (5) edge node[below] {$1 \textcolor{ACMRed}{- 5\varepsilon}$} (5);
    \draw[->, loop below] (6) edge node[below] {$1 \textcolor{ACMRed}{- 4\varepsilon}$} (6);
    \draw[->, loop below] (3) edge node[below] {$1$} (3);
    \draw[->, ACMRed] (4) edge node[below] {$\varepsilon$} (3);
    \draw[->, loop below] (1) edge node[below] {$1$} (1);
    \end{tikzpicture}
\caption{The portion of the Markov chain $\chain{P_\varepsilon}$ reachable from $(s, t)$, induced by a policy $P_\varepsilon$ that is optimal for all $\varepsilon \in [0, \frac{1}{8}]$.  We denote the limit policy $P_0 \in \mathcal{P}_\tau$ by $P$.}
\label{fig:perturbation-policy}
\end{center}
\end{subcaptionblock}
\caption{An illustration for the proof of \cref{theorem:main-continuity}.}
\label{fig:illustration}
\end{figure}

\begin{proof}[Proof of \cref{theorem:cav-conjecture}]
By \cref{theorem:cav-theorem}, it is sufficient to show that for all $s$, $t \in S$, if $s \sim t$ and $s \not\simeq t$ then $\delta_{\_}(s, t) : (S \to \Dreal(S)) \to [0, 1]$ is discontinuous at $\tau$.

Let $s$, $t \in S$.  Assume that $s \sim t$ and $s \not\simeq t$.  As $s \sim t$, by Theorem~\ref{theorem:distance-zero}, we have $\delta_\tau (s, t) = 0$.  Since $s \not\simeq t$, for all policies $P \in \mathcal{P}$, we have $\reachdelta{P}(s, t) < 1 = 1 - \delta_\tau (s, t)$.  Hence, by \cref{theorem:main-continuity}b, $\delta_{\_}(s, t) : (S \to \Dreal(S)) \to [0, 1]$ is discontinuous at~$\tau$.
\end{proof}

\begin{example}
\label{example:distance-jump}
We revisit our running example, the family of labelled Markov chains in Figure~\ref{fig:lmc}, to investigate how the bisimilarity distance changes when the transition function varies; see \cref{fig:distance-jump-example} for an illustration.

Consider the pair of states $(a_1, b_2)$.  The policy $P_0 \in \mathcal{P}_\tau$, partially illustrated in Figure~\ref{fig:policy}, is optimal for $(a_1, b_2)$.  Hence, we have $\delta_\tau(a_1, b_2) = \frac{1}{4}$.  Observe that $\reachdelta{P_0}(a_1, b_2) = \frac{1}{4}$ and no policy can improve the probability of reaching $S^2_\Delta$ from $(a_1, b_2)$.  Thus, $\delta_{\_}(a_1, b_2) : (S \to \Dreal(S)) \to [0, 1]$ is discontinuous at $\tau$.  Indeed, we have $\delta_{\tau_n}(a_1, b_2) = \frac{3}{4}$ for all $n \in \nat$, so small changes in the transition probabilities cause the bisimilarity distance to jump up by $\frac{1}{2}$.

In contrast, consider the pair of states $(a_1, c_1)$.  The policy $P_0 \in \mathcal{P}_\tau$ is also optimal for $(a_1, c_1)$, hence  we have $\delta_\tau(a_1, c_1) = \frac{1}{2}$.  Since $\reachdelta{P_0}(a_1, c_1) = \frac{1}{2} = 1 - \delta_\tau(a_1, c_1)$, we know that $\delta_{\_}(a_1, c_1) : (S \to \Dreal(S)) \to [0, 1]$ is continuous at $\tau$.  Note that $\lim_n \delta_{\tau_n}(a_1, c_1) = \lim_n  \left( \frac{1}{2} + \frac{1}{4n} \right) = \frac{1}{2}$.  Moreover, for all sequences $(\sigma_n)_n$ converging to $\tau$, we have that $\lim_n \delta_{\sigma_n}(a_1, c_1) = \frac{1}{2}$.\lipicsEnd
\end{example}

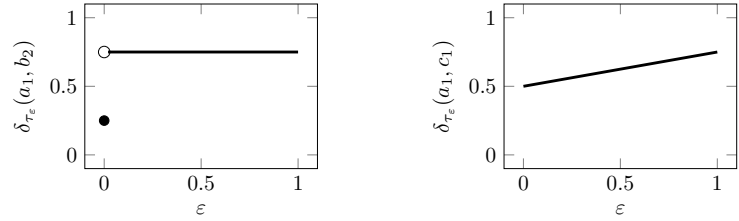
\begin{figure}[ht!]
\begin{center}
    \begin{tikzpicture}[scale=0.9, every node/.style={scale=0.9}]
    \begin{axis}
    [xlabel={$\varepsilon$}, ylabel={$\delta_{\tau_\varepsilon} (a_1, c_1)$}, width=5cm, height=4cm, at={(0.72\linewidth,0)}, ymin=-0.1,ymax=1.1]
    \addplot[domain=0:1,line width=1.2pt] {0.5 + x/4};
    \end{axis}
    \begin{axis}
    [xlabel={$\varepsilon$}, ylabel={$\delta_{\tau_\varepsilon} (a_1, b_2)$}, width=5cm, height=4cm, at={(0.28\linewidth,0)}, ymin=-0.1,ymax=1.1]
    \addplot[mark=*, mark options={scale=1}] coordinates {(0,0.25)};
    \addplot[mark=o, mark options={scale=1.2}] coordinates {(0,0.75)};
    \addplot[domain=0.02:1,line width=1.2pt] {0.75};
    \end{axis}
    \end{tikzpicture}
\end{center}
\caption{An illustration of the distances between the state pairs $(a_1, b_2)$ and $(a_1, c_1)$ from the labelled Markov chain in \cref{fig:running-example}.  The pairs of states have distance $\frac{1}{4}$ and $\frac{1}{2}$, respectively, when $\varepsilon = 0$.  For all policies $P \in \mathcal{P}_\tau$, $\reachdelta{P}(a_1, b_2) \leq \frac{1}{4} < \frac{3}{4} = 1 - \delta_{\tau}(a_1, b_2)$.  In contrast, there exists a policy $P \in \mathcal{P}_\tau$ such that $\reachdelta{P}(a_1, c_1) = \frac{1}{2} = 1 - \delta_{\tau}(a_1, c_1)$.  Hence, the distance is continuous at $\tau$ for $(a_1, c_1)$ only.}
\label{fig:distance-jump-example}
\end{figure}

Bisimilarity and robust bisimilarity coincide if and only if the bisimilarity distance is continuous for all pairs of states.

\begin{corollary}
\label{corollary:continuity-trivial}
$\mathord{\sim} = \robust \Longleftrightarrow$ for all $s$, $t \in S$, $\delta_{\_}(s, t) : (S \to \Dreal(S)) \to [0, 1]$ is continuous at $\tau$.
\end{corollary}

%% file: text/algorithm-continuity.tex
\section{Algorithm for Deciding Continuity}
\label{section:algorithm-continuity}

Theorem~\ref{theorem:bisim-continuity-poly-time} establishes polynomial time decidability of continuity for bisimilar pairs of states.
In this section, we extend this result.
We show that one can decide in polynomial time whether $\delta_\tau(s, t)$ is continuous regardless of whether $s$ and $t$ are bisimilar.


To this end, we first introduce two auxiliary functions, which allow us to express the algorithm as a fixed point computation.

Let $R \subseteq S \times S$. We use the notation $\overline{R} = (S \times S) \setminus R$.
We define a function $A : 2^{S \times S} \to 2^{S \times S}$ that, intuitively, collects those pairs of states that are incompatible with the relation $R$.
More precisely, for states $s$ and $t$, the pair $(s, t)$ is included in $A(R)$ either if it is already known to violate robustness, or if every optimal way of matching the transitions of $s$ and $t$ necessarily results in a matched pair that lies outside $R$.
Formally,
\[
A(R) = (\mathord{\sim} \setminus \robust) \cup \{\, (s, t) \in S^2_{?,\tau} \mid \forall \omega \in \Coptimal(\tau(s), \tau(t)) : \support(\omega) \cap R \neq \varnothing \,\}.
\]
Dually, we define a function $B : 2^{S \times S} \to 2^{S \times S}$ that, intuitively, identifies those pairs of states that are compatible with $R$. That is, $(s, t)$ belongs to $B(R)$ either if it is already known to be robust, or if there exists an optimal way of matching the transitions of $s$ and $t$ such that all matched pairs are contained in $R$.
Formally, $B(R) = \overline{A(\overline{R})}$
or, equivalently,
\[
B(R) = \robust \cup S^2_1 \cup \{\, (s, t) \in S^2_{?,\tau} \mid \exists \omega \in \Coptimal(\tau(s), \tau(t)) : \support(\omega) \subseteq R \,\}.
\]

\begin{proposition}
\label{proposition:monotone}
$A$ and $B$ are monotone with respect to $\subseteq$.
\end{proposition}

Therefore, $A$ and $B$ are monotone functions from the complete lattice $2^{S \times S}$ to itself.  According to the Knaster-Tarski fixed point theorem, $A$ has a least fixed point, which we denote by $\alpha$ and $B$ has a greatest fixed point, which we denote by $\beta$.

The following two propositions state that $\beta$ is the set of pairs of states for which the distance function is continuous at $\tau$, while $\alpha$ is its complement.

\begin{proposition}
\label{proposition:AB-opposites}
$\beta = \overline{\alpha}$.
\end{proposition}
\begin{proof}
First we show that for all $R \subseteq S \times S$, $R$ is a fixed point of $A$ if and only if $\overline{R}$ is a fixed point of $B$. Let $R \subseteq S \times S$.
\[
A(R) = R 
~\Longleftrightarrow~ \overline{A(R)} = \overline{R}
~\Longleftrightarrow~ \overline{A(\overline{\overline{R}})} = \overline{R}
~\Longleftrightarrow~ B(\overline{R}) = \overline{R}
\]

We now prove the two inclusions.  Since $\beta$ is a fixed point of $B$, we can conclude that $\overline{\beta}$ is a fixed point of $A$.  Thus, $\alpha \subseteq \overline{\beta}$, as $\alpha$ is the least fixed point of $A$.  It follows that $\beta \subseteq \overline{\alpha}$.

Similarly, since $\alpha$ is a fixed point of $A$, we can conclude that $\overline{\alpha}$ is a fixed point of $B$.  Thus, $\overline{\alpha} \subseteq \beta$, as $\beta$ is the greatest fixed point of $B$.  Hence, $\beta = \overline{\alpha}$.
\end{proof}

\begin{proposition}
\label{proposition:beta-continuous}
$\beta = \{\, (s, t) \in S \times S \mid \delta_{\_}(s, t) \mbox{ is continuous at } \tau \,\}$.
\end{proposition}
\begin{proof}[Proof Sketch]
Let $s$, $t \in S$.  We prove the two implications.

Suppose that $(s, t) \in \beta$.  Using the definition of $B$, we construct a policy $P$ such that $\reachdelta{P}(s, t) = 1 - \delta_\tau (s, t)$.  Thus, by \cref{theorem:main-continuity}, $\delta_{\_}(s, t)$ is continuous at $\tau$.

Suppose that $\delta_{\_}(s, t)$ is continuous at $\tau$.  Then, by \cref{theorem:main-continuity}, there exists a policy $P$ such that $\reachdelta{P}(s, t) = 1 - \delta_\tau (s, t)$.
Towards a contradiction, assume that $(s, t) \not\in \beta$. It follows from Proposition~\ref{proposition:AB-opposites} that $(s, t) \in \alpha$.  Then, using the definition of $A$, we show that $(s, t)$ reaches $\mathord{\sim} \setminus \robust$ in $\chain{P}$.  It follows that $\reachdelta{P}(s, t) < 1 - \delta_\tau (s, t)$, a contradiction.
\end{proof}

We now introduce the algorithm to decide continuity in polynomial time.
Let $T \subseteq S \times S$.  We use the following notation below:
$\mathrm{Pre}(T) = \{\, (s, t) \in S \times S \mid (\support(\tau(s)) \times \support(\tau(t))) \cap T \neq \varnothing \,\}$.

\begin{algorithm}[ht]
\caption{Continuity}
\label{algorithm:continuity}
\DontPrintSemicolon
\SetAlgoNoLine
\KwIn{A labelled Markov chain $\langle S, L, \tau, \ell \rangle$, the bisimilarity distance function $\delta_\tau : S \times S \to [0, 1]$, the set of bisimilar states $\sim$, and the set of robustly bisimilar states $\simeq$}
\KwOut{$\{\, (s, t) \in S \times S \mid \delta_{\_}(s, t)$ is continuous at $\tau \,\}$}
$D_{\mathrm{old}} \gets \varnothing$\;
$D \gets \mathord{\sim} \setminus \robust$\;
\Repeat{$D = D_{\mathrm{old}}$}{
  $D_{\mathrm{old}} \gets D$\;
  \ForEach{$(s, t) \in (S^2_{?,\tau} \cap \mathrm{Pre}(D_{\mathrm{old}})) \setminus D_{\mathrm{old}}$}{
    \If{$\forall \omega \in \Coptimal(\tau(s), \tau(t))$ we have $\support(\omega) \cap D_{\mathrm{old}} \neq \varnothing$}{
      $D \gets D \cup \{(s, t)\}$\;
    }
  }
}
\Return $(S \times S) \setminus D$
\end{algorithm}

\begin{proposition}
\label{proposition:algorithm-alpha}
Algorithm~\ref{algorithm:continuity} computes the set $D = \alpha$.
\end{proposition}
\begin{proof}[Proof Sketch]
We show the following loop invariant of Algorithm~\ref{algorithm:continuity}: $D = A(D_{\mathrm{old}}) \subseteq \alpha$.

The loop terminates when a fixed point is reached, therefore, by the loop invariant we know that $D = A(D) \subseteq \alpha$. Thus, $D$ is a fixed point of $A$ less than or equal to $\alpha$. Since $\alpha$ is the least fixed point of $A$, we can conclude that Algorithm~\ref{algorithm:continuity} computes $\alpha$.
\end{proof}

\begin{theorem}
Algorithm~\ref{algorithm:continuity} returns the set $\{\, (s, t) \in S \times S \mid \delta_{\_}(s, t)$ is continuous at $\tau \,\}$.
\end{theorem}
\begin{proof}
This is immediate from \cref{proposition:AB-opposites,proposition:beta-continuous,proposition:algorithm-alpha}.
\end{proof}

In the remainder of this section, we show that Algorithm~\ref{algorithm:continuity} runs in polynomial time.  Note that $\sim$, $\simeq$, and $\delta_\tau$ can be computed in polynomial time \cite{CBW12, FKPB25c}.

\begin{proposition}
\label{proposition:min-cost-max-flow}
For all $s$, $t \in S$ and $X \subseteq S \times S$, we can determine in polynomial time whether there exists $\omega \in \Coptimal(\tau(s), \tau(t))$ such that $\support(\omega) \subseteq X$.
\end{proposition}
\begin{proof}[Proof Sketch]
We reduce the problem to a minimum cost maximum flow instance \cite{FF62}.

We construct the network, with source $s$ and sink $t$, as follows.  For each $x \in S$, we add two vertices $(x, 0)$ and $(x, 1)$, an edge from $s$ to $(x, 0)$ with capacity $\tau(s)(x)$, and an edge from $(x, 1)$ to $t$ with capacity $\tau(t)(x)$.  These edges have cost $0$.  For each $(x, y) \in X$, we add an edge from $(x, 0)$ to $(y, 1)$ with capacity $1$ and cost $\delta_\tau(x, y)$.

Let $f : E \to \mathbb{R}$ denote the flow on the edges.  The value of the flow, $|f|$ is defined as the net flow entering the sink, or equivalently, the net flow exiting the source.  The cost of the flow is defined as the sum of the cost per unit flow over all edges.  We can use standard network flow algorithms to solve the problem. For example, Dinic's maximum flow algorithm maximizes the value of the flow in polynomial time \cite{D06} and the cycle-cancelling algorithm then minimizes the cost of the flow, without changing the value of the flow \cite{K67}, in polynomial time \cite{GT89}.  Note that the problem can also be solved in polynomial time using linear programming.

We show that there exists $\omega \in \Coptimal(\tau(s), \tau(t))$ such that $\support(\omega) \subseteq X$ if an only if the resulting value of the flow is $1$ and the cost of the flow is $\delta_\tau(s, t)$.
Intuitively, a flow of value $1$ corresponds exactly to a coupling $\omega \in \Creal(\tau(s), \tau(t))$, with $\omega(x, y) = f((x,0), (y,1))$, and the construction of the network ensures that $\support(\omega) \subseteq X$.  Conversely, any such coupling defines a feasible flow of value $1$.  It follows that the cost of the flow equals $\sum_{x \in S} \sum_{y \in S} \delta_\tau(x, y)\, \omega(x, y)$.  Hence, the cost of the flow is $\delta_\tau(s, t)$ if and only if $\omega$ is an optimal coupling.
\end{proof}

\begin{proposition}
Algorithm~\ref{algorithm:continuity} runs in polynomial time.
\end{proposition}
\begin{proof}
Let $|S| = m$. Since at least one pair of states is added to $D$ in each iteration, the algorithm requires at most $m^2$ iterations.  Furthermore, during each iteration, at most $m^2$ pairs of states are checked in the inner loop.

The check on line~6 is equivalent to determining whether there exists $\omega \in \Coptimal(\tau(s), \tau(t))$ such that $\support(\omega) \subseteq \overline{D_{\mathrm{old}}}$.  By Proposition~\ref{proposition:min-cost-max-flow}, this can be done in polynomial time.
\end{proof}

%% file: text/experiments.tex
\section{Experiments}
\label{section:experiments}

To evaluate the efficiency and practical usefulness of our algorithm to decide continuity,
we implemented it in the widely used probabilistic model checker PRISM~\cite{KNP11},
an open-source tool for quantitative verification and analysis of probabilistic models, including labelled Markov chains.

\subsection{Implementation}

Our algorithm requires as input the set of bisimilar states $\sim$, the set of robustly bisimilar states $\simeq$, and the bisimilarity distance function $\delta_\tau : S \times S \to [0, 1]$.  We first briefly mention how they are computed.

PRISM's implementation of the traditional (i.e., non-robust) bisimilarity algorithm is a standard partition-refinement approach
which uses the signature-based method of Derisavi \cite{D07}.
The initial partition is based on the labelling of the states.  Let $\Pi$ be the current partition and $E_\Pi$ be the set of equivalence classes in $\Pi$.  Then the new partition is computed as $\{\, (s, t) \in \Pi \mid \forall B \in E_\Pi : \tau(s)(B) = \tau(t)(B) \,\}$.
Each state and equivalence class (referred to as a block) is represented by an integer ID.  The current partition of the state space is tracked by an array that is indexed by state IDs and contains the corresponding block IDs.

PRISM's implementation of the robust bisimilarity algorithm was described in \cite{FKPB25c}.  Given a policy $P \in \mathcal{P}$, a set $R \subseteq S \times S$ \emph{supports} a path $(u_1, v_1) \hdots (u_n, v_n)$ in $\chain{P}$ if for all $1 \leq i \leq n$ we have $(u_i, v_i) \in R$ and $\support(P(u_i, v_i)) \subseteq R$.
Let $\Pi$ be the current partition.  Then the new partition is the largest bisimulation such that for every $(s, t) \in \Pi$ there exists a policy $P$ such that $\Pi$ supports a path from $(s, t)$ to $S^2_\Delta$ in $\chain{P}$.

An algorithm to compute the bisimilarity distance, based on Condon's simple policy iteration algorithm \cite{C90}, was proposed in \cite{BBLM13, TB16, TB18}.
Distance zero (i.e. bisimilarity) and one are computed first.
Then, starting from an arbitrary vertex policy, the algorithm repeatedly improves the policy with a locally optimal vertex coupling, until no further local improvement is possible.
For the bisimilarity distance, Tang and van Breugel \cite{TB17} showed that, in practice, policy iteration approaches are more efficient than the theoretically superior linear programming algorithm.
Hence, we adapt an implementation of the aforementioned algorithm, provided in \cite{T18}.
We rely on Google's OR-Tools \cite{ortools} to compute optimal couplings.
The OR-Tools suite is an open source software for combinatorial optimization and includes solvers for linear programming.

We implemented Algorithm~\ref{algorithm:continuity} in Java as part of PRISM's explicit-state model checking engine.\footnote{\href{https://github.com/zainabfatmi/prism/blob/concur/prism/src/explicit/bisim/distances/Continuity.java}{github.com/zainabfatmi/prism/blob/concur/prism/src/explicit/bisim/distances/Continuity.java}}  The transition probabilities of the labelled Markov chain are stored as a two-dimensional array and the set $D$ of discontinuous pairs is represented as a flat boolean array where a state pair $(s, t)$ is indexed by $s \cdot |S| + t$.  Throughout the algorithm, we exploit the symmetry of $D$ by setting both $(s, t)$ and $(t, s)$ simultaneously.  It is initialised to ${\sim} \setminus {\simeq}$, that is, those pairs of states that are bisimilar but not robustly so.

The main loop iterates over all pairs $(s, t) \notin D$ with non-zero distance and matching labels, restricting attention to those that are predecessors of $D$.  For each such candidate pair, we attempt to find an optimal coupling of $\tau(s)$ and $\tau(t)$ whose support lies entirely outside the known discontinuous pairs $D$, as described in the proof of \cref{proposition:min-cost-max-flow}, using Google's OR-Tools.  If no such coupling exists, the pair is added to $D$.  This process continues until no new pairs are added to $D$, at which point a fixed point has been reached.

\subsection{Experimental Setup}
All experiments were run on a MacBook with an M1 chip and 16GB memory.

We evaluated our algorithm by applying it to all (discrete-time) labelled Markov chains from the Quantitative Verification Benchmark Set (QVBS) \cite{HKPQR19} for which the bisimilarity distance algorithm could complete within the available memory.
QVBS is a comprehensive collection of probabilistic models which is designed as a benchmark suite for quantitative verification and analysis tools and is the foundation of the Quantitative Verification Competition (QComp), which compares the performance, versatility, and usability of such tools.

For an additional source of models, we also use jpf-probabilistic~\cite{FCDWTB21}.
Java PathFinder (JPF) \cite{VHBPL03} is a popular model checker for Java code, and the JPF extension jpf-probabilistic provides Java implementations of sixty randomized algorithms \cite{FCDWTB21}. As shown in \cite{FCDWTB21}, JPF, extended by jpf-probabilistic and jpf-label \cite{F20}, can be used in tandem with PRISM to check properties of these algorithms and supplement JPF’s qualitative results with quantitative information.  A description of the subset of these algorithms utilized in our study is provided in \ifthenelse{\isundefined{\techReport}}{\cite[Appendix~H]{arxiv}}{Appendix~\ref{appendix:experiments}}.

\subsection{Results}
Table~\ref{table:results} reports the results for benchmarks where the minimized models obtained via traditional bisimilarity and robust bisimilarity differ.
These are of particular interest because, according to \cref{corollary:continuity-trivial}, they are instances where our algorithm identifies pairs of states for which the distance is discontinuous.
Note that the property used for each benchmark determines the labelling of the model.
In the table, \emph{Time} denotes the runtime (in seconds) and \emph{\%} denotes the percentage of state pairs for which the distance is continuous. The percentage of continuous pairs with distance less than one is given in parentheses.

\begin{table}[ht!]
\caption{The results of the benchmarks for which continuity is non-trivial.}
\rowcolors{1}{gray!10}{white}
\centering
\begin{tabular}{ m{0.2\textwidth} m{0.14\textwidth} R{0.07\textwidth} R{0.14\textwidth} R{0.13\textwidth} R{0.14\textwidth} }
  \toprule
  \rowcolor{white}
  \multicolumn{3}{c}{Benchmark} & Distance & \multicolumn{2}{c}{\qquad Continuity} \\
  \midrule
  Name (Property) & Parameters & States & Time (s) & Time (s) & \% \\ 
  \midrule
  brp (p1) & N=2, M=2 & 89 & 17.276 & 16.164 & 34.6\phantom{0} (6.5) \\
  & N=2, M=3 & 116 & 58.826 & 78.724 & 27.5\phantom{0} (5.0) \\
  & N=2, M=4 & 143 & 156.344 & 252.931 & 22.9\phantom{0} (4.2) \\
  brp (p4) & N=2, M=2 & 89 & 2.396 & 2.042 & 15.8\phantom{0} (3.3) \\
  & N=2, M=3 & 116 & 7.987 & 7.834 & 12.2\phantom{0} (2.4) \\
  & N=2, M=4 & 143 & 20.332 & 19.645 & 9.9\phantom{0} (1.9) \\
  & N=2, M=5 & 170 & 40.691 & 38.109 & 8.4\phantom{0} (1.6) \\
  & N=4, M=2 & 173 & 21.533 & 18.225 & 8.3\phantom{0} (1.7) \\
  crowds (positive) & R=2, S=2 & 77 & 126.911 & 8.450 & 16.7\phantom{0} (4.9) \\
  oscillators (power) & T=6, N=3 & 57 & 5.059 & 0.002 & 94.6 (82.2) \\
  oscillators (time) & T=6, N=3 & 57 & 4.687 & 0.002 & 94.6 (82.2) \\
  \bottomrule
\end{tabular}
\label{table:results}
\end{table}

In terms of the size of the state space, we are limited by what the bisimilarity distance algorithm can handle, constrained by 16GB of memory.  This is similar to the model sizes considered in \cite{TB17}.
Nevertheless, the results are promising: the additional cost of deciding continuity is, on all but one benchmark, less than that for computing the bisimilarity distance and in some cases substantially so, demonstrating that our approach is feasible in practice.
Moreover, our results suggest that when the bisimilarity distance is continuous for most state pairs, the continuity check is particularly efficient.

\begin{table}[htp]
\caption{The results of some benchmarks for which continuity is trivial.}
\rowcolors{1}{gray!10}{white}
\centering
\begin{tabular}{ m{0.31\textwidth} m{0.14\textwidth} R{0.07\textwidth} R{0.11\textwidth} m{0.07\textwidth} R{0.12\textwidth} }
  \toprule
  \rowcolor{white}
  \multicolumn{3}{c}{Benchmark} & \multicolumn{2}{c}{\qquad Distance} & Continuity \\
  \midrule
  Name (Property) & Parameters & States & Time (s) & Trivial & Time (s) \\ 
  \midrule
  fair biased coin (heads) & p=0.9 & 13 & 0.161 & no & 0.000 \\
  herman (steps) & N=5 & 32 & 7.015 & no & 0.001 \\
  leader-sync (elected) & N=3, K=4 & 147 & 0.988 & yes & 0.004\\
  majority element (incorrect) & s=50, t=5 & 101 & 26.608 & no & 0.011 \\ 
  haddad-monmege (target) & N=20, p=0.7 & 41 & 312.192 & no & 0.001 \\
  pollards factorization (input) & i=4004 & 5 & 0.000 & yes & 0.000 \\
  queens (success) & n=4 & 16 & 0.001 & yes & 0.000 \\
  \bottomrule
\end{tabular}
\label{table:trivial-results}
\end{table}

Table~\ref{table:trivial-results} presents the largest feasible model per benchmark where the minimized models obtained by traditional bisimilarity and robust bisimilarity coincide.
In these cases, by \cref{corollary:continuity-trivial}, the distance is continuous for all pairs of states and our algorithm terminates almost immediately. The column \emph{Trivial} denotes whether the distance is an integer, that is $0$ or $1$, for all state pairs.  Evidently, this does not affect the time to decide continuity.

Across all benchmarks, the experiments demonstrate that the additional computation for deciding continuity is comparable to or faster than calculating the bisimilarity distance.
Hence the extra step of deciding continuity is reasonable in practice, while providing additional robustness guarantees.
As discussed in the introduction, in applications where transition
probabilities are subject to approximation or uncertainty, continuity guarantees that sufficiently small errors in these values lead only to small changes in the bisimilarity distance. Conversely, discontinuity identifies situations in which arbitrarily small perturbations may cause a non-negligible change in the distance, highlighting that conclusions drawn from the computed distance should be interpreted with caution.

%% file: text/conclusion.tex
\section{Conclusions and Future Work}
\label{section:conclusion}

The concept of robust bisimilarity was proposed in \cite{FKPB25c} to address the lack of robustness of bisimilarity, where it was shown to be a sufficient condition for continuity.
In this paper, we have strengthened that result by proving that robust bisimilarity is also a necessary condition for continuity.
We have completed the theoretical characterization of robustness in \cref{theorem:main-continuity}.
Moreover, we have shown that continuity is decidable in polynomial time for all pairs of states and we have extended the experimental study done in \cite{FKPB25c}, with particular emphasis on the quantitative setting.

As future work, we aim to provide a bound on how much the distance can increase to, under small perturbations of the transition probabilities.
Additionally, we plan to investigate the continuity of the distance function when we restrict to perturbations that preserve the underlying graph structure, as this is often known, even when the transition probabilities are subject to approximation.

%% file: text/appendix.tex
\section{Metric Topology}
\label{appendix:metric-topology}

\begin{definition}
The function $d_{E} : [0, 1] \times [0, 1] \to [0, 1]$ is defined by
\[
d_E(r, s) = | r - s |.
\]
\end{definition}

In the remainder, let $(Y, d)$ be a 1-bounded metric space.  According to, for example, \cite[Theorem~4.1.17]{E89}, a set $A \subseteq Y$ is \emph{compact} if and only if it is \emph{sequentially compact}, that is, each sequence in $A$ has a converging subsequence.  Given a sequence $(y_n)_n$, we denote a subsequence by $(y_{f(n)})_n$, where the function $f : \nat \to \nat$ is strictly increasing.

\begin{proposition}[{\cite[Example~3.1.25]{E89}}]
\label{proposition:unit-interval-compact}
$([0, 1], d_E)$ is compact.
\end{proposition}

\begin{definition}
Given a metric $d$ on $Y$, the function $d_F : (X \to Y) \times (X \to Y) \to [0, 1]$ is defined by
\[
d_F(f, g) = \max_{x \in X} d(f(x), g(x)).
\]
\end{definition}

\begin{proposition}[{\cite[Theorem~3.2.4]{E89}}]
\label{proposition:finite-product-compact}
If $(Y, d)$ is compact then $(X \to Y, d_F)$ is compact.
\end{proposition}

Let $(Z, d')$ be a 1-bounded metric space.  According to, for example, \cite[page~250]{E89}, a function $f : Y \to Z$ is \emph{continuous at $y \in Y$} if for any sequence $(y_n)_n$ in $Y$ converging to $y$, the sequence $(f(y_n))_n$ converges to $f(y)$, that is, $\lim_n f(y_n) = f(y)$.  Note that $y = \lim_n y_n$ is defined in term of $d$, whereas $\lim_n f(y_n)$ is defined in terms of $d'$.  A function $f : Y \to Z$ is \emph{continuous} if it is continuous at each $y \in Y$.  For $c > 0$, a function $f : Y \to Z$ is \emph{$c$-Lipschitz} if for all $y$, $z \in Y$, $d'(f(y), f(z)) \leq c \, d(y, z)$.  A function $f : Y \to Z$ is \emph{nonexpansive} if it is 1-Lipschitz.

\begin{proposition}
\label{proposition:Lipschitz-continuous}
Each Lipschitz function is continuous.
\end{proposition}
\begin{proof}
Let $c > 0$ and $f : Y \to Z$ be $c$-Lipschitz.  According to, for example, \cite[Proposition~4.1.8]{E89}, to conclude that $f$ is continuous it suffices to show that 
\[
\forall \alpha > 0 : \exists \beta > 0 : \forall y, z \in Y : d(y, z) < \beta \Rightarrow d'(f(y), f(z)) < \alpha.
\]
Take $\beta = \frac{\alpha}{c}$.
\end{proof}

\begin{definition}
The function $d_{TV} : \Dreal(X) \times \Dreal(X) \to [0, 1]$ is defined by
\[
d_{TV}(\mu, \nu) = \max_{x \in X} | \mu(x) - \nu(x) |.
\]
\end{definition}

\begin{proposition}
The function $f : Y \to [0, 1]$ is continuous at $y \in Y$ if and only if $f$ is lower semi-continuous at $y$ and upper semi-continuous at $y$.
\end{proposition}
\begin{proof}
Let $f : Y \to [0, 1]$ and $y \in Y$.  We prove two implications.  Assume that $f$ is lower and upper semi-continuous at $y$.  Let $(y_n)_n$ be a sequence in $Y$ converging to $y$.  Since
\begin{align*}
f(y) 
& \leq \lim \inf_n f(y_n)
\proofcomment{$f$ is lower semi-continuous at $y$}\\
& \leq \lim \sup_n f(y_n)\\
& \leq f(y)
\proofcomment{$f$ is upper semi-continuous at $y$}
\end{align*}
we can conclude that $\lim \inf_n f(y_n) = \lim \sup_n f(y_n) = f(y)$.  Hence, $\lim_n f(y_n) = f(y)$.

Assume that $f$ is continuous at $y$.  Let $(y_n)_n$ be a sequence in $Y$ converging to $y$.  Since the sequence $(f(y_n))_n$ converges, we have that $\lim \inf_n f(y_n) = \lim \sup_n f(y_n) = \lim_n f(y_n) = f(y)$.  Hence, $f$ is lower and upper semi-continuous at $y$.
\end{proof}

\section{Value Function}

\begin{definition}
The function $\Gamma : (S \times S \to \Dreal(S \times S)) \to (S \times S \to [0, 1]) \to (S \times S \to [0, 1])$ is defined by
\[
\Gamma_P(d)(s, t) = \left \{
\begin{array}{ll}
1
& \hspace{0.5cm} \mbox{if $(s, t) \in S^2_1$}\\
P(s, t) \cdot d
& \hspace{0.5cm} \mbox{otherwise,}	
\end{array}
\right .
\]
where $P(s, t) \cdot d = \sum_{u, v \in S} P(s, t)(u, v) \; d(u, v)$.
\end{definition}

\begin{proposition}
\label{proposition:gamma-Lipschitz}
For all $P$, $Q : S \times S \to \Dreal(S \times S)$, $d_F(\Gamma_P, \Gamma_Q) \leq |S|^2\, d_F(P, Q)$.
\end{proposition}
\begin{proof}
Let $P$, $Q : S \times S \to \Dreal(S \times S)$.  It suffices to show that for all $d : S \times S \to [0, 1]$ and $s$, $t \in S$, we have that $|\Gamma_P(d)(s, t) - \Gamma_Q(d)(s, t)| \leq |S|^2\, d_F(P, Q)$.  Let $d : S \times S \to [0, 1]$ and $s$, $t \in S$.  We distinguish the following cases.
\begin{itemize}
\item
If $(s, t) \in S^2_1$ then
\[
|\Gamma_P(d)(s, t) - \Gamma_Q(d)(s, t)| 
= |1 - 1|
= 0
\leq |S|^2\, d_F(P, Q).
\]
\item 
Otherwise,
\begin{align*}
|\Gamma_P(d)(s, t) - \Gamma_Q(d)(s, t)|
& = | P(s, t) \cdot d - Q(s, t) \cdot d|\\
& = | (P(s, t) - Q(s, t)) \cdot d|\\
& \leq \sum_{u, v \in S} |P(s, t)(u, v) - Q(s, t)(u, v)| \, d(u, v)\\
& \leq \sum_{u, v \in S} |P(s, t)(u, v) - Q(s, t)(u, v)| \\
& \leq |S|^2 \, |P(s, t) - Q(s, t)|\\
& \leq  |S|^2 \, d_F(P, Q). \qedhere
\end{align*}
\end{itemize}
\end{proof}

\begin{proposition}
\label{proposition:gamma-nonexpansive}
For all $P : S \times S \to \Dreal(S \times S)$ and $d$, $e : S \times S \to [0, 1]$, $d_F(\Gamma_P(d), \Gamma_P(e)) \leq d_F(d, e)$.
\end{proposition}
\begin{proof}
Let $P : S \times S \to \Dreal(S \times S)$ and $d$, $e : S \times S \to [0, 1]$.  It suffices to show that for all $s$, $t \in S$, we have that $|\Gamma_P(d)(s, t) - \Gamma_P(e)(s, t)| \leq d_F(d, e)$.  Let $s$, $t \in S$.  We distinguish the following cases.
\begin{itemize}
\item
If $(s, t) \in S^2_1$ then
\[
|\Gamma_P(d)(s, t) - \Gamma_P(e)(s, t)| 
= |1 - 1|
= 0
\leq d_F(d, e).
\]
\item 
Otherwise,
\begin{align*}
|\Gamma_P(d)(s, t) - \Gamma_P(e)(s, t)|
& = |P(s, t) \cdot d - P(s, t) \cdot e|\\
& = |P(s, t) \cdot (d - e)|\\
& \leq \sum_{u, v \in S} P(s, t)(u, v) \, |d(u, v) - e(u, v)|\\
& \leq \sum_{u, v \in S} P(s, t)(u, v) \, d_F(d, e)\\
& = d_F(d, e). \qedhere
\end{align*}
\end{itemize}  
\end{proof}

\begin{corollary}
\label{corollary:Gamma-continuous}
The function $\Gamma_\cdot(\cdot) : ((S \times S \to \Dreal(S \times S)) \times (S \times S \to [0, 1])) \to (S \times S \to [0, 1])$ is continuous.
\end{corollary}
\begin{proof}
Let $P$, $Q : S \times S \to \Dreal(S \times S)$ and $d$, $e : S \times S \to [0, 1]$.  Then
\begin{align*}
d_F(\Gamma_P(d), \Gamma_Q(e))
& \leq d_F(\Gamma_P(d), \Gamma_Q(d)) + d_F(\Gamma_Q(d), \Gamma_Q(e)) \proofcomment{triangle inequality}\\
& \leq |S|^2\, d_F(P, Q) + d_F(d, e)
\proofcomment{\cref{proposition:gamma-Lipschitz,proposition:gamma-nonexpansive}}\\
& \leq 2 |S|^2\, \max \{ d_F(P, Q), d_F(d, e) \}.
\end{align*}
Hence, the function $\Gamma_\cdot(\cdot)$ is $2 |S|^2$-Lipschitz and, by Proposition~\ref{proposition:Lipschitz-continuous}, continuous.
\end{proof}

For each $P : S \times S \to \Dreal(S \times S)$, $\Gamma_P$ is a monotone function from the complete lattice $S \times S \to [0, 1]$ to itself (see, for example, \cite[Proposition~6.1.3]{T18}).  According to the Knaster-Tarski fixed point theorem, $\Gamma_P$ has a least fixed point, which we denote by $\reachone{P}$.  Note that $\chain{P}$ is a Markov chain.

Recall that $\mathcal{P}_\tau$ is the set of policies for $\tau$ and that the subscript $\tau$ is omitted when clear from the context.

\begin{theorem}[{\cite[Theorem~10.15]{BK08}}]
\label{theorem:reach-s1}
For all $P \in \mathcal{P}$ and $s$, $t \in S$, $\reachone{P}(s, t)$ is the probability of reaching $S^2_1$ from $(s, t)$ in $\chain{P}$.
\end{theorem}

\begin{theorem}[{\cite[Theorem~8]{CBW12}}]
\label{proposition:minimal-coupling}
$\displaystyle
\delta_\tau = \min_{P \in \mathcal{P}} \reachone{P}.
$
\end{theorem}

The above theorem is proved by showing that $\delta_\tau \sqsubseteq \reachone{P}$ for all $P \in \mathcal{P}$ and that there exists $P \in \mathcal{P}$ such that $\delta_\tau = \reachone{P}$.

\begin{definition}
\label{definition:optimal-policy}
Let $s$, $t \in S$.  
\begin{itemize}
\item 
A policy $P \in \mathcal{P}$ is \emph{optimal for} $(s, t)$ if $\reachone{P}(s, t) = \delta_\tau(s, t)$.
\item
A policy $P \in \mathcal{P}$ is \emph{optimal} if for all $s$, $t \in S$, $P$ is optimal for $(s, t)$.
\end{itemize}
\end{definition}

Note that from Theorem~\ref{proposition:minimal-coupling} we can conclude that optimal policies exist.

\begin{proposition}[{\cite[Proposition~13]{FKPB25}}]
\label{proposition:optimal-characterization}
For all $P \in \mathcal{P}$, the following are equivalent.
\begin{enumerate}
\item 
$P$ is optimal
\item 
$\Gamma_P(\delta_\tau) = \delta_\tau$
\item 
$\Gamma_P(\delta_\tau) \sqsubseteq \delta_\tau$
\end{enumerate}
\end{proposition}

\begin{remark}
\label{remark:optimal-couplings}
Let $P \in \mathcal{P}$ be an optimal policy.  By Proposition~\ref{proposition:optimal-characterization}, $\Gamma_P(\delta_\tau) = \delta_\tau$.  Hence $\delta_\tau(s, t) = P(s, t) \cdot \delta_\tau$.  Let $\omega = P(s, t) \in \Creal(\tau(s), \tau(t))$.  Then, $\delta_\tau(s, t) = \omega \cdot \delta_\tau$.  Thus, we can conclude that optimal couplings exist.
\end{remark}

\begin{proposition}
\label{proposition:reachable-optimal}
For all $s$, $t$, $u$, $v \in S$, if $P \in \mathcal{P}$ is optimal for $(s, t)$ and $(s, t)$ can reach $(u, v)$ in $\chain{P}$ then $P$ is optimal for $(u, v)$.
\end{proposition}
\begin{proof}
Let $s$, $t$, $u$, $v \in S$.  Assume that $P \in \mathcal{P}$ is optimal for $(s, t)$ and $(s, t)$ can reach $(u, v)$ in $\chain{P}$.  We prove this proposition by induction on the length of a shortest path from $(s, t)$ to $(u, v)$.  In the base case, $(s, t) = (u, v)$ and the proposition is vacuously true.  In the inductive case, let $(s, t) (w, x) \ldots (u, v)$ be a shortest path from $(s, t)$ to $(u, v)$.  Then $(s, t) \in S^2_\Delta \cup S^2_{0?}$ and $P(s, t)(w, x) > 0$.  Towards a contradiction, assume that $P$ is not optimal for $(w, x)$.  Then $\delta_\tau(w, x) < \reachone{P}(w, x)$ by Theorem~\ref{proposition:minimal-coupling}.  We also have that $\delta_\tau \sqsubseteq \reachone{P}$, again by Theorem~\ref{proposition:minimal-coupling}.  Hence,
\begin{align*}
\delta_\tau(s, t)
& = \reachone{P}(s, t)
\proofcomment{$P$ is optimal for $(s,t)$}\\
& = \Gamma_P(\reachone{P})(s, t)\\
& = P(s, t) \cdot \reachone{P}
\proofcomment{$(s, t) \in S^2_\Delta \cup S^2_{0?}$}\\
& > P(s, t) \cdot \delta_\tau \proofcomment{$P(s, t)(w, x) > 0$, $\delta_\tau(w, x) < \reachone{P}(w, x)$ and $\delta_\tau \sqsubseteq \reachone{P}$}\\
& \geq \inf_{\omega \in \Creal(\tau(s), \tau(t))} \omega \cdot \delta_\tau
\proofcomment{$P(s, t) \in \Creal(\tau(s), \tau(t))$}\\
& = \Delta_\tau(\delta_\tau)(s, t)\\
& = \delta_\tau(s, t),
\end{align*}
a contradiction.  Therefore, $P$ is optimal for $(w, x)$.  By the induction hypothesis, $P$ is optimal for $(u, v)$.
\end{proof}

\section{Linear Algebra}

We denote the infinity norm by $\| \cdot \|$.  Recall that for an $n$-vector $x$, we have that $\| x \| = \max_{0 \leq i < n} |x_i|$ and for an $m \times n$-matrix $A$, we have that $\| A \| = \max_{0 \leq i < m} \sum_{0 \leq j < n} | A_{ij} |$.  Given $n$-vectors $x$ and $y$, we write $x \lneqq y$ if $x_i \leq y_i$ for all $0 \leq i < n$ and $x_j  < y_j$ for some $0 \leq j < n$.  We denote constant vectors and matrices simply by their value.  A matrix $A$ is \emph{strictly substochastic} if $A 1 \lneqq 1$.  

An $n \times n$ matrix $A$ is \emph{cogredient} to a matrix $E$ if for some permutation matrix $P$, $PAP^t = E$. $A$ is \emph{reducible} if it is cogredient to $E = \big[\begin{smallmatrix}
B & 0 \\
C & D
\end{smallmatrix}\big]$, where $B$ and $D$ are square matrices, or if $n = 1$ and $A = 0$. Otherwise, A is \emph{irreducible}.
In the following, we will only rely on the characterization of irreducibility given in the following theorem.

\begin{theorem}[{\cite[Theorem~2.2.1]{BP94}}]
\label{theorem:irreducible}
A nonnegative $n \times n$-matrix $A$ is irreducible if and only if for every  $0 \leq i, j < n$ there exists $m > 0$ such that $A_{ij}^m > 0$.
\end{theorem}

\begin{proposition}[{\cite[Proposition~16]{FKPB25}}]
\label{proposition:irreducible-invertible}
Let $A$ be an irreducible and strictly substochastic $n \times n$-matrix.  Then $I - A$ is invertible.
\end{proposition}

\section{Continuity}

We say that $P \in \mathcal{P}$ is an \emph{$S^2_\Delta$-closed policy} if 
\[
\forall s \in S : \support(P(s,s)) \subseteq S^2_\Delta.
\]

\begin{proposition}[{\cite[Proposition~11]{FKPB25}}]
\label{proposition:policy-Lipschitz}
For all $\sigma$, $\tau : S \to \Dreal(S)$, and $S^2_\Delta$-closed policies $P \in \mathcal{P}_\sigma$, there exists an $S^2_\Delta$-closed policy $Q \in \mathcal{P}_\tau$ such that $d_F(P, Q) \leq 2\, d_F(\sigma, \tau)$.
\end{proposition}

\begin{proposition}[{\cite[Proposition~14]{FKPB25}}]
\label{proposition:closed-communication-class}
Let $P \in \mathcal{P}$ be an $S^2_\Delta$-closed policy.  If $C$ is a closed communication class of $\chain{P}$ then
\begin{enumerate}
\item 
$C = \{ (s, t) \}$ for some $(s, t) \in S^2_1$, or
\item 
$C \subseteq S^2_\Delta$, or
\item 
$C \subseteq S^2_{0,\tau}$.
\end{enumerate}
\end{proposition}

\begin{proposition}[{\cite[Proposition~15]{FKPB25}}]
\label{proposition:delta-closed-policy}
Let $s$, $t \in S$.  If there exists a policy $P \in \mathcal{P}$ such that $\reachdelta{P}(s, t) = p$, then there exists an $S^2_\Delta$-closed policy $Q \in \mathcal{P}$ such that $\reachdelta{Q}(s, t) = p$.
\end{proposition}

\begin{corollary}
\label{corollary:optimal-closed-policy}
Let $s$, $t \in S$.  If there exists a policy $P \in \mathcal{P}$ such that $\reachdelta{P}(s, t) = 1 - \delta_\tau (s, t)$, then there exists an $S^2_\Delta$-closed policy $Q \in \mathcal{P}$ that is optimal for $(s, t)$ such that $\reachdelta{Q}(s, t) = 1 - \delta_\tau (s, t)$.
\end{corollary}
\begin{proof}
Let $s$, $t \in S$.  Assume that $P \in \mathcal{P}_\tau$ is a policy such that $\reachdelta{P}(s, t) = 1 - \delta_\tau (s, t)$.  By Proposition~\ref{proposition:delta-closed-policy} there exists an $S^2_\Delta$-closed policy $Q \in \mathcal{P}_\tau$ such that $\reachdelta{Q}(s, t) = 1 - \delta_\tau (s, t)$.  It follows from Theorem~\ref{theorem:reach-s1} and Proposition~\ref{proposition:closed-communication-class} that $\reachdelta{Q}(s, t) + \reachone{Q}(s, t) \leq 1$.  Hence, we have $\reachone{Q}(s, t) \leq \delta_\tau(s, t)$.  According to Theorem~\ref{proposition:minimal-coupling}, $\reachone{Q}(s, t) \geq \delta_\tau(s, t)$.  Therefore, $\reachone{Q}(s, t) = \delta_\tau(s, t)$ and, thus, $Q$ is optimal for $(s, t)$ due to Definition~\ref{definition:optimal-policy}.
\end{proof}

\begin{proposition}
\label{proposition:phi-impossible}
For all $s$, $t \in S$, for all $P \in \mathcal{P}$, we have $\reachdelta{P}(s, t) \leq 1 - \delta_\tau (s, t)$.
\end{proposition}
\begin{proof}
Let $s$, $t \in S$.  Towards a contradiction, assume that there exists a policy $P \in \mathcal{P}$ such that $\reachdelta{P}(s, t) > 1 - \delta_\tau (s, t)$.  Then, by Proposition~\ref{proposition:delta-closed-policy}, there exists an $S^2_\Delta$-closed policy $Q \in \mathcal{P}$ such that $\reachdelta{Q}(s, t) > 1 - \delta_\tau (s, t)$.  It follows from Theorem~\ref{theorem:reach-s1} and Proposition~\ref{proposition:closed-communication-class} that $\reachone{Q} (s, t) < \delta_\tau (s, t)$.  This contradicts Theorem~\ref{proposition:minimal-coupling}.
\end{proof}

\begin{lemma}
\label{lemma:continuous}
Let $s$, $t \in S$.  If there exists a policy $P \in \mathcal{P}_\tau$ such that $\reachdelta{P}(s, t) = 1 - \delta_\tau (s, t)$, then the function $\delta_{\_}(s, t) : (S \to \Dreal(S)) \to [0, 1]$ is upper semi-continuous at $\tau$.
\end{lemma}
\begin{proof}
Let $s$, $t \in S$.   Assume that $(\tau_n)_n$ is a sequence in $S \to \Dreal(S)$ that converges to $\tau$.  It suffices to show that $\lim \sup_n \delta_{\tau_n}(s, t) \leq \delta_{\tau}(s, t)$.

Assume that $Q \in \mathcal{P}_\tau$ is a policy such that $\reachdelta{Q}(s, t) = 1 - \delta_\tau (s, t)$.  By Corollary~\ref{corollary:optimal-closed-policy} there exists an $S^2_\Delta$-closed policy $P \in \mathcal{P}_\tau$ that is optimal for $(s, t)$ such that $\reachdelta{P}(s, t) = 1 - \delta_\tau (s, t)$.  According to Proposition~\ref{proposition:policy-Lipschitz}, for each $n \in \nat$, there exists an $S^2_\Delta$-closed policy $P_n \in \mathcal{P}_{\tau_n}$ such that $d_F(P_n, P) \leq 2\, d_F(\tau_n, \tau)$.  Hence, $(P_n)_n$ converges to $P$.

Consider the directed graph consisting of the communication classes of $\chain{P}$ reachable from $(s, t)$ as vertices.  There is an edge from communication class $C$ to communication class $D$ if there exist $(u, v) \in C$ and $(w, x) \in D$ such that $P(u, v)(w, x) > 0$.  This graph is acyclic.  We first prove that for all communication classes $C$ of $\chain{P}$ that are reachable from $(s, t)$ and for all $(u, v) \in C$,
\begin{equation}
\label{equation:lim-gamma-is-gamma}
\lim_n \reachone{P_n}(u, v) = \reachone{P}(u, v)
\end{equation} 
by induction on the length of a longest path from $C$ in the communication classes graph.

In the base case, the length of a longest path in the communication classes graph is one.  Here we consider the \emph{closed} communication classes $C$.  Since $P$ is optimal for $(s, t)$, $(s, t)$ reaches $S^2_1$ with probability $\delta_\tau(s, t)$ in $\chain{P}$.  Then, as $P$ is an $S^2_\Delta$-closed policy and $\reachdelta{P}(s, t) = 1 - \delta_\tau (s, t)$, by Proposition~\ref{proposition:closed-communication-class}, we have either $C \subseteq S^2_\Delta$ or $C = \{ (u, v)\}$ for some $(u, v) \in S^2_1$.

\begin{itemize}
\item 
Suppose $C \subseteq S^2_\Delta$.  Let $(u, v) \in C$.  According to Theorem~\ref{theorem:reach-s1}, for all $n \in \nat$, $\reachone{P_n}(u, v) = 0$ and $\reachone{P}(u, v) = 0$, as they are $S^2_\Delta$-closed policies.  Therefore, (\ref{equation:lim-gamma-is-gamma}).
\item 
Suppose $C = \{ (u, v)\}$ and $(u, v) \in S^2_1$.  Then for all $n \in \nat$, 
\[
\reachone{P_n}(u, v) 
= \Gamma_{P_n}(\reachone{P_n})(u, v)
= 1. 
\]
Similarly, we can conclude that $\reachone{P}(u, v) = 1$.  Hence, (\ref{equation:lim-gamma-is-gamma}).
\end{itemize}

Next, we consider the inductive case.  Let $C$ be a communication class of $\chain{P}$ reachable from $(s, t)$.  Let $B$ be the set of state pairs of all communication classes that can be reached from $C$ in the communication classes graph via a path of length greater than $1$.  By induction, for all $(u, v) \in B$, (\ref{equation:lim-gamma-is-gamma}) holds.  Let $A$ be the set of all other state pairs, that is, $A = (S \times S) \setminus (B \cup C)$.
\begin{center}
\begin{tikzpicture}[xscale=1.9]
\node[state] (st) at (1, 3.5) {$s, t$};

\node[ellipse, draw, fit={(0.5, 1.5) (1, 2) (1, 1) (1.5,1.5)}, fill=ACMBlue!30] {};

\node[state] (uv) at (1, 1.8) {$u, v$};

\node at (2, 1.5) {$C$};

\draw[->,decorate,decoration={snake, pre length=1pt, post length=2pt, amplitude=1.3pt}] (st) to (uv);

\node[ellipse, draw, fit={(-0.5, -1) (0, -1.5) (0.5, -1) (0,-0.5)}, fill=ACMOrange!40] {};
\node[ellipse, draw, fit={(1.5, -1) (2, -1.5) (2.5, -1) (2,-0.5)}, fill=ACMOrange!40] {};
\node[ellipse, draw, fit={(0.5, -2.5) (1, -2) (1, -3) (1.5,-2.5)}, fill=ACMOrange!40] {};

\node[draw, fit={(-1, -3.5) (-1, 0) (3, 0) (3,-3.5)}] {};

\node at (3.25, -1.75) {$B$};

\draw[->]  (0.5, 1.5) to (0, -0.5);
\draw[->]  (1.5, 1.5) to (2, -0.5);
\draw[->]  (0, -1.5) to (0.5, -2.5);
\draw[->]  (2, -1.5) to (1.5, -2.5);
\end{tikzpicture}
\end{center}

For $X \subseteq S \times S$ and $n \in \nat$, consider the vectors
\begin{align*}
\reachone{P_n, X} & = (\reachone{P_n}(u, v))_{(u, v) \in X}\\
\reachone{P, X} & = (\reachone{P}(u, v))_{(u, v) \in X}
\end{align*}
and the matrices
\begin{align*}
P_n^X & = (P_n(u, v)(w, x))_{(u, v) \in C, (w, x) \in X}\\
P^X & = (P(u, v)(w, x))_{(u, v) \in C, (w, x) \in X}
\end{align*}
For all $n \in \nat$, $\reachone{P_n} = \Gamma_{P_n}(\reachone{P_n})$ and, hence,
\[
\reachone{P_n, C} = P_n^C \reachone{P_n, C} + P_n^B \reachone{P_n, B} + P_n^A \reachone{P_n, A}.
\]
Since $\reachone{P} = \Gamma_P(\reachone{P}) $, we also have
\[
\reachone{P, C} = P^C \reachone{P, C} + P^B \reachone{P, B} + P^A \reachone{P, A}.
\]
From the communication classes graph we can infer that for all $(u, v) \in C$, $\support(P(u, v)) \subseteq B \cup C$.  Hence, $P^A = 0$.  Since $\lim_n P_n = P$, we have that
\begin{align}
\lim_n P_n^A & = P^A = 0 \nonumber\\
\lim_n P_n^B & = P^B
\label{equation:limit-matrix}\\
\lim_n P_n^C & = P^C \nonumber
\end{align}

Next, we prove that the inverse $(I - P^C)^{-1}$ exists.  We distinguish the following cases.
\begin{itemize}
\item 
If $P^C = 0$ then $I - P^C = I$, which has an inverse.
\item
Otherwise, $P^C \gneqq 0$.  For every $(u, v)$, $(w, x) \in C$ there exists $m$ such that $(P^C)^m_{(u, v), (w, x)} > 0$, because $C$ is a communication class.  By Theorem~\ref{theorem:irreducible}, $P^C$ is irreducible.  Since the communication class $C$ is not closed, $P^C$ is strictly substochastic.  Hence, by Proposition~\ref{proposition:irreducible-invertible}, the inverse $(I - P^C)^{-1}$ exists.
\end{itemize}

Therefore,
\begin{align*}
& \reachone{P_n, C} - \reachone{P, C}\\
& = (P_n^C \reachone{P_n, C} + P_n^B \reachone{P_n, B} + P_n^A \reachone{P_n, A}) - (P^C \reachone{P, C} + P^B \reachone{P, B} + P^A \reachone{P, A})\\
& = (P_n^C \reachone{P_n, C} + P_n^B \reachone{P_n, B} + P_n^A \reachone{P_n, A}) - (P^C \reachone{P, C} + P^B \reachone{P, B})
\proofcomment{$P^A = 0$}\\
& = P^C (\reachone{P_n, C} - \reachone{P, C}) + (P_n^C - P^C) \reachone{P_n, C} + P^B (\reachone{P_n, B} - \reachone{P, B}) + (P_n^B - P^B) \reachone{P_n, B}\\
& \qquad + P_n^A \reachone{P_n, A}.
\end{align*}
Hence,
\begin{align*}
&(I - P^C) \, (\reachone{P_n, C} - \reachone{P, C})\\
&= (P_n^C - P^C) \reachone{P_n, C} + P^B (\reachone{P_n, B} - \reachone{P, B}) + (P_n^B - P^B) \reachone{P_n, B} + P_n^A \reachone{P_n, A}.
\end{align*}
As a consequence,
\begin{align*}
&\reachone{P_n, C} - \reachone{P, C}\\
&= (I - P^C)^{-1} \, ( (P_n^C - P^C) \reachone{P_n, C} + P^B (\reachone{P_n, B} - \reachone{P, B}) + (P_n^B - P^B) \reachone{P_n, B} + P_n^A \reachone{P_n, A}).
\end{align*}
Hence,
\begin{align*}
& \| \reachone{P_n, C} - \reachone{P, C} \|\\
& = \| (I - P^C)^{-1} \, ( (P_n^C - P^C) \reachone{P_n, C} + P^B (\reachone{P_n, B} - \reachone{P, B}) + (P_n^B - P^B) \reachone{P_n, B}\\
& \qquad + P_n^A \reachone{P_n, A}) \|\\
& \leq \| (I - P^C)^{-1} \| \, (\| P_n^C - P^C \| \, \|\reachone{P_n, C}\| + \| P^B \| \, \| \reachone{P_n, B} - \reachone{P, B} \| + \| P_n^B - P^B \| \, \| \reachone{P_n, B} \|\\
& \qquad + \| P_n^A \| \, \| \reachone{P_n, A} \|)\\
& \leq \| (I - P^C)^{-1} \| \, (\| P_n^C - P^C \| + \| P^B \| \, \| \reachone{P_n, B} - \reachone{P, B} \| + \| P_n^B - P^B \| + \| P_n^A \|)\\
& \quad \proofcomment{$\| \reachone{P_n, X} \| \leq 1$}\\
& \leq \| (I - P^C)^{-1} \| \, (\| P_n^C - P^C \| + |S|^2 \, \| \reachone{P_n, B} - \reachone{P, B} \| + \| P_n^B - P^B \| + \| P_n^A \|)\\
& \quad \proofcomment{$\| P^B \| \leq |S|^2$}
\end{align*}

We need to prove that $\lim_n \reachone{P_n, C} = \reachone{P, C}$ and that this limit exists.  It is sufficient to show that $\limsup_n \|  \reachone{P_n, C} - \reachone{P, C} \| = 0$.  From the above we can conclude that this holds, as
\begin{align*}
& \limsup_n \| \reachone{P_n, C} - \reachone{P, C} \|\\
& \leq \limsup_n \| (I - P^C)^{-1} \| \, (\| P_n^C - P^C \| + |S|^2 \, \| \reachone{P_n, B} - \reachone{P, B} \| + \| P_n^B - P^B \| + \| P_n^A \|)\\
& = \| (I - P^C)^{-1} \| \, (\limsup_n \| P_n^C - P^C \| + |S|^2 \, \limsup_n \| \reachone{P_n, B} - \reachone{P, B} \|\\
& \qquad + \limsup_n \| P_n^B - P^B \| + \limsup_n \| P_n^A \|)\\
& = \| (I - P^C)^{-1} \| \, (|S|^2 \, \limsup_n \| \reachone{P_n, B} - \reachone{P, B} \|)
\proofcomment{(\ref{equation:limit-matrix})}\\
& = 0
\proofcomment{$\limsup_n \| \reachone{P_n, B} - \reachone{P, B} \| = 0$ by induction}
\end{align*}
This proves (\ref{equation:lim-gamma-is-gamma}).

Assume that $(s, t)$ belongs to communication class $C$.  Then
\begin{align*}
\limsup_n \delta_{\tau_n}(s, t)
& \leq \limsup_n \reachone{P_n}(s, t)
\proofcomment{$\delta_{\tau_n} \sqsubseteq \reachone{P_n}$ by Theorem~\ref{proposition:minimal-coupling}}\\
& = \limsup_n \reachone{P_n, C}(s, t)
\proofcomment{$(s, t) \in C$}\\
& = \reachone{P, C}(s, t)
\proofcomment{(\ref{equation:lim-gamma-is-gamma})}\\
& = \reachone{P}(s, t)
\proofcomment{$(s, t) \in C$}\\
& = \delta_\tau(s, t)
\proofcomment{$P \in \mathcal{P}_\tau$ is optimal for $(s, t)$}
\end{align*}
Hence, the function $\delta_{\_}(s, t)$ is upper semi-continuous at $\tau$.
\end{proof}

\begin{proposition}
\label{proposition:continuous}
For all $s$, $t \in S$, if there exists a policy $P \in \mathcal{P}$ such that $(s, t)$ reaches $S^2_\Delta$ with probability $1 - \delta_\tau (s, t)$ in $\chain{P}$ then the function $\delta_{\_}(s, t) : (S \to \Dreal(S)) \to [0, 1]$ is continuous at $\tau$.
\end{proposition}
\begin{proof}
Follows directly from Proposition~\ref{proposition:lower-semi-continuous} and Lemma~\ref{lemma:continuous}.
\end{proof}

\section{Couplings}

\begin{definition}
Let $\omega \in \Sreal(X \times X)$.  The \emph{support graph} of $\omega$ is $(V, E)$ where
\begin{itemize}
\item 
$V = X \times \{ 0, 1 \}$ and
\item 
$E = \{\, \{ (x, 0), (y, 1) \} \mid \omega(x, y) > 0 \,\}$.
\end{itemize}
\end{definition}

We denote the vertices of the convex polytope $\Creal(\mu, \nu)$ by $V(\Creal(\mu, \nu))$.

\begin{proposition}
\label{propositions:couplings-nonempty}
For all $\mu$, $\nu \in \Dreal(X)$, $\Creal(\mu, \nu)$ is nonempty.
\end{proposition}
\begin{proof}
Let $\mu$, $\nu \in \Dreal(X)$. The North-West corner method \cite{H41} constructs a coupling of $\mu$ and $\nu$.
\end{proof}

\begin{proposition}[{\cite[Theorem~4]{KW67}}]
\label{proposition:forest}
For all $\mu$, $\nu \in \Dreal(X)$ and $\omega \in \Creal(\mu, \nu)$, $\omega$ is a vertex if and only if the support graph of $\omega$ is a forest.
\end{proposition}



\input{text/discontinuity}

\section{Algorithm}

\begin{proof}[Proof of Proposition~\ref{proposition:monotone}]
Let $X$, $Y \subseteq S \times S$, with $X \subseteq Y$.  Let $s$, $t \in S$.

Assume that $(s,t) \in A(X)$.  If $(s,t) \in \mathord{\sim} \setminus \robust$, then $(s,t) \in A(Y)$.  Otherwise, $(s,t) \in S^2_{?,\tau}$.  Then for all $\omega \in \Coptimal(\tau(s), \tau(t))$, we have $\support(\omega) \cap X \neq \varnothing$.  Since $X \subseteq Y$, for all $\omega \in \Coptimal(\tau(s), \tau(t))$, we have $\support(\omega) \cap Y \neq \varnothing$.  Thus, $(s,t) \in A(Y)$.

Assume that $(s,t) \in B(X)$.  If $(s,t) \in \robust \cup S^2_1$, then $(s,t) \in B(Y)$.  Otherwise, $(s,t) \in S^2_{?,\tau}$.  Then there exists $\omega \in \Coptimal(\tau(s), \tau(t))$ such that $\support(\omega) \subseteq X \subseteq Y$.  Thus, $(s,t) \in B(Y)$.
\end{proof}

\begin{proposition}
\label{proposition:alpha-discontinuous}
For all $s$, $t \in S$, if $(s, t) \in \alpha$, then for all policies $P \in \mathcal{P}$ that are optimal for $(s, t)$, $(s, t)$ reaches $\mathord{\sim} \setminus \robust$ in $\chain{P}$.
\end{proposition}
\begin{proof}
Let $s$, $t \in S$.
We show the following loop invariant of Algorithm~\ref{algorithm:continuity}: $\forall (s, t) \in D, \forall P \in \mathcal{P}$ that are optimal for $(s, t) : \exists$ path of length $\leq n$ from $(s, t)$ to $\mathord{\sim} \setminus \robust$ in $\chain{P}$, where $n$ is initialised to $0$ and incremented at the end of each loop.

Initially, $D = \mathord{\sim} \setminus \robust$.
For all $(s, t) \in \mathord{\sim} \setminus \robust$, for all $P \in \mathcal{P}$, $(s, t)$ can reach $\mathord{\sim} \setminus \robust$ in $0$ steps in $\chain{P}$.  Hence, the loop invariant holds before the loop.

Assume that the loop invariant holds before an iteration of the loop.  Then for all $(s, t) \in D_{\mathrm{old}}$, for all $P \in \mathcal{P}$ that are optimal for $(s, t)$, there exists a path of length $\leq n$ from $(s, t)$ to $\mathord{\sim} \setminus \robust$ in $\chain{P}$.

Assume that $(s, t)$ is added to $D$ on line~7.  Then $(s, t) \in (S^2_{?,\tau} \cap \mathrm{Pre}(D_{\mathrm{old}})) \setminus D_{\mathrm{old}}$ and for all $\omega \in \Coptimal(\tau(s), \tau(t))$, $\support(\omega) \cap D_{\mathrm{old}} \neq \varnothing$.  Let $P \in \mathcal{P}$ be an optimal policy for $(s, t)$.  Then $P(s, t) \in \Creal(\tau(s), \tau(t))$ is an optimal coupling.  Thus, there exists $(u, v) \in \support(P(s, t)) \cap D_{\mathrm{old}}$.  By Proposition~\ref{proposition:reachable-optimal}, $P$ is an optimal policy for $(u, v)$.  By the induction hypothesis, there exists a path of length $\leq n$ from $(u, v)$ to $\mathord{\sim} \setminus \robust$ in $\chain{P}$.  Therefore, there exists a path of length $\leq n + 1$ from $(s, t)$ to $\mathord{\sim} \setminus \robust$ in $\chain{P}$.

Hence, for all $(s, t) \in D$, for all $P \in \mathcal{P}$ that are optimal for $(s, t)$, there exists a path of length $\leq n + 1$ from $(s, t)$ to $\mathord{\sim} \setminus \robust$ in $\chain{P}$.  Thus, the loop invariant is maintained in each iteration of the loop.
\end{proof}



\begin{proof}[Proof of Proposition~\ref{proposition:beta-continuous}]
Let $s$, $t$, $u$, $v \in S$.  We prove the two implications.

Assume that $(s, t) \in \beta$.  For all $(u, v) \in \beta \cap S^2_{?,\tau}$, there exists $\omega_{uv,\beta} \in \Coptimal(\tau(u), \tau(v))$ such that $\support(\omega_{uv,\beta}) \subseteq \beta$.  By Proposition~\ref{proposition:maximal-support-coupling}, for all $(u, v) \in \robust$, there exists a maximal $\robust$-support coupling $\omega_{uv,\robust} \in \Creal(\tau(u), \tau(v))$.  Let $Q \in \mathcal{P}$ be an $S^2_\Delta$-closed policy. We define $P \in \mathcal{P}$ as follows.
\[
P(s, t) = \left \{
\begin{array}{ll}
\omega_{st,\beta}
& \hspace{0.5cm} \mbox{if $(s, t) \in \beta \cap S^2_{?,\tau}$}\\
\omega_{st,\robust}
& \hspace{0.5cm} \mbox{if $(s, t) \in \robust \setminus S^2_\Delta$}\\
Q(s, t)
& \hspace{0.5cm} \mbox{otherwise}.
\end{array}
\right .
\]
Observe that $P$ is optimal for $(s, t)$.  Thus, by  Theorem~\ref{theorem:reach-s1}, $(s, t)$ reaches $S^2_1$ with probability $\delta_\tau (s, t)$ in $\chain{P}$.  Moreover, for all $(u, v) \in \robust$, $(u, v)$ reaches $S^2_\Delta$ and $(s, t)$ cannot reach $S^2_{0,\tau} \setminus \robust$.  It follows from Proposition~\ref{proposition:ccc} that $(s, t)$ reaches $S^2_\Delta$ with probability $1 - \delta_\tau (s, t)$ in $\chain{P}$.  Thus, by \cref{theorem:main-continuity}, $\delta_{\_}(s, t)$ is continuous at $\tau$.

Assume that $\delta_{\_}(s, t)$ is continuous at $\tau$.  Then, by \cref{theorem:main-continuity}, there exists a policy $P \in \mathcal{P}$ such that $(s, t)$ reaches $S^2_\Delta$ with probability $1 - \delta_\tau(s, t)$ in $\chain{P}$.  By Corollary~\ref{corollary:optimal-closed-policy} there exists an $S^2_\Delta$-closed policy $Q$ that is optimal for $(s, t)$ such that $(s, t)$ reaches $S^2_\Delta$ with probability $1 - \delta_\tau(s, t)$ in $\chain{Q}$.
Towards a contradiction, assume that $(s, t) \not\in \beta$. It follows from Proposition~\ref{proposition:AB-opposites} that $Q(s, t) \cap \alpha \neq \varnothing$.  Then, by Proposition~\ref{proposition:alpha-discontinuous}, $(s, t)$ reaches some $(u, v) \in \mathord{\sim} \setminus \robust$ in $\chain{Q}$.  Note that $(u, v)$ reaches $S^2_\Delta$ with probability $< 1$ in $\chain{Q}$.  Hence, $(s, t)$ reaches $S^2_\Delta$ with probability $< 1 - \delta_\tau(s, t)$ in $\chain{Q}$, a contradiction.
\end{proof}

\begin{proof}[Proof of Proposition~\ref{proposition:algorithm-alpha}]
We show the following loop invariants of Algorithm~\ref{algorithm:continuity}: $D = A(D_{\mathrm{old}}) \subseteq \alpha$ and $D_{\mathrm{old}} \subseteq D$.

Initially, $D_{\mathrm{old}} = \varnothing$ and $D = \mathord{\sim} \setminus \robust$.  Thus, $D = A(D_{\mathrm{old}}) \subseteq \alpha$ and $D_{\mathrm{old}} \subseteq D$.  Hence, the loop invariant holds before the loop.

Assume that the loop invariant holds before an iteration of the loop, then $D = A(D_{\mathrm{old}}) \subseteq \alpha$ and $D_{\mathrm{old}} \subseteq D$.  During the loop, $D_{\mathrm{old}}' = D$ and $D' = D \cup \{\, (s, t) \in S^2_{?,\tau} \cap \mathrm{Pre}(D_{\mathrm{old}}) \mid \forall \omega \in \Coptimal(\tau(s), \tau(t)) : \support(\omega) \cap D_{\mathrm{old}}' \neq \varnothing \,\}$.  Clearly $D_{\mathrm{old}}' \subseteq D'$.  We need to show that $D' = A(D_{\mathrm{old}}')$.  We prove the two implications.
Let $(s, t) \in D'$.  We distinguish the following two cases.
\begin{itemize}
    \item Suppose $(s, t) \in D$.  Then $(s, t) \in A(D_{\mathrm{old}})$. Since $D_{\mathrm{old}} \subseteq D = D_{\mathrm{old}}'$ and, by Proposition~\ref{proposition:monotone}, $A$ is monotone, we have $(s, t) \in A(D_{\mathrm{old}}')$.
    \item Suppose $(s, t) \not\in D$.  Then $(s, t) \in S^2_{?,\tau} \cap \mathrm{Pre}(D_{\mathrm{old}})$ and for all $\omega \in \Coptimal(\tau(s), \tau(t))$, we have $\support(\omega) \cap D_{\mathrm{old}}' \neq \varnothing$.  Thus, $(s, t) \in A(D_{\mathrm{old}}')$.
\end{itemize}
To prove the other implication, let $(s, t) \in A(D_{\mathrm{old}}')$.  We distinguish three cases.
\begin{itemize}
    \item Suppose $(s, t) \in \mathord{\sim} \setminus \robust$.  Since $\mathord{\sim} \setminus \robust \subseteq D \subseteq D'$, we have $(s, t) \in D'$.
    \item Suppose $(s, t) \in S^2_{?,\tau} \cap D$.  Then $(s, t) \in D'$.
    \item Suppose $(s, t) \in S^2_{?,\tau} \setminus D$.  Then, for all $\omega \in \Coptimal(\tau(s), \tau(t))$, we have $\support(\omega) \cap D_{\mathrm{old}}' \neq \varnothing$.  It follows that $(s, t) \in \mathrm{Pre}(D_{\mathrm{old}})$.  Hence, $(s, t) \in D'$.
\end{itemize}
Thus, the loop invariant is maintained in each iteration of the loop.

The loop terminates when a fixed point is reached, therefore, by the loop invariant we know that $D = A(D) \subseteq \alpha$. Thus, $D$ is a fixed point of $A$ less than or equal to $\alpha$. Since $\alpha$ is the least fixed point of $A$, we can conclude that Algorithm~\ref{algorithm:continuity} computes $\alpha$.
\end{proof}

\begin{proof}[Proof of \cref{proposition:min-cost-max-flow}]
Recall the definition of the maximum flow problem \cite{FF62}.  Let $\mathbb{R}_{\geq 0} = \{ x \in \mathbb{R} \mid x \geq 0 \}$.  Given a capacitated directed graph $G = (V, E, c)$, where $c : E \to \mathbb{R}_{\geq 0}$ is a non-negative capacity function on edges, with distinguished source and sink vertices $s$ and $t$, respectively.  A flow $f : E \to \mathbb{R}$ is defined as a function on the edges satisfying the following constraints: 
\begin{enumerate}[(a)]
    \item $\forall e \in E: 0 \leq f(e) \leq c(e)$, the flow of an edge cannot exceed its capacity and
    \item $\forall v \in V \setminus \{ s, t \}$, we have \[\sum_{u \in V : (v, u) \in E} f(v, u) = \sum_{u \in V : (u, v) \in E} f(u, v),\] flow is conserved at all vertices except the source and sink, that is, the sum of flows entering a vertex must equal the sum of flows exiting that vertex.
\end{enumerate}
The value of the flow is defined as the net flow entering the sink, or equivalently, the net flow exiting the source: 
\[
|f| = \sum_{u \in V : (s, u) \in E} f(s, u) = \sum_{v \in V : (v, t) \in E} f(v, t).
\]
The problem is to maximize the value of the flow.

Let $s$, $t \in S$ and $X \subseteq S \times S$.  We construct the network as follows, with $s$ as the source and $t$ as the sink.  For each $x \in S$, we add two vertices $(x, 0)$ and $(x, 1)$, an edge from $s$ to $(x, 0)$ with capacity $\tau(s)(x)$, and an edge from $(x, 1)$ to $t$ with capacity $\tau(t)(x)$.  Let $x$, $y \in S$.  If $(x, y) \in X$, we add an edge from $(x, 0)$ to $(y, 1)$ with capacity $1$.  Note that $\forall f : E \to \mathbb{R} : |f| \leq 1$.  We then use Dinic's maximum flow algorithm, which maximizes the value of the flow and runs in $\mathcal{O}(|S|^2 |X|)$ time \cite{D06}.

Additionally, we introduce a non-negative cost function on edges $k : E \to \mathbb{R}_{\geq 0}$ and assign costs as follows:
\[
k(e) = \left \{
\begin{array}{ll}
\delta_\tau (x, y)
& \hspace{0.5cm} \mbox{if $e = ((x, 0), (y, 1))$}\\
0
& \hspace{0.5cm} \mbox{otherwise.}
\end{array}
\right .
\]
The cost of the flow is defined as the sum of the cost per unit flow over all edges:
\[
|k_f| = \sum_{e \in E} k(e)f(e) = \sum_{x \in X} \sum_{y \in X} k((x, 0), (y, 1)) f((x, 0), (y, 1)).
\]
Note that $\forall f : E \to \mathbb{R} : |k_f| \leq |f|$.
We then run the cycle-cancelling algorithm, which minimizes the cost of the flow, without changing the value of the flow \cite{K67} and runs in polynomial time \cite{GT89}.

It suffices to show that $\exists f : E \to \mathbb{R} : |f| = 1$ and $|k_f| = \delta_\tau(s, t) \Longleftrightarrow \exists \omega \in \Coptimal(\tau(s), \tau(t)) : \support(\omega) \subseteq X$. We prove the two implications.

Assume that there exists a flow $f : E \to \mathbb{R}$ such that $|f| = 1$ and $|k_f| = \delta_\tau(s, t)$.  Since the sum of the capacities of the edges entering the sink $t$ is equal to $\tau(t)(S) = 1$, by constraint (a), $\forall x \in S : f((x, 1), t) = \tau(t)(x)$.  Then, constraint (b) specifies that $\forall x \in S : \sum_{y \in S} f((y, 0), (x, 1)) = f((x, 1), t) = \tau(t)(x)$.  Similarly, $\forall x \in S : \sum_{y \in S} f((x, 0), (y, 1)) = f(s, (x, 0)) = \tau(s)(x)$.  Therefore, the resulting flow is a coupling $\omega \in \Creal(\tau(s), \tau(t))$.  By construction, $((x, 0), (y, 1)) \in E$ if and only if $(x, y) \in X$, thus, $\support(\omega) \subseteq X$.  Since $|k_f| = \delta_\tau(s, t)$, by the definition of the cost of the flow, we know that $\omega$ is an optimal coupling.

To prove the other implication, assume that there exists $\omega \in \Coptimal(\tau(s), \tau(t))$ such that $\support(\omega) \subseteq X$.  Then, $\forall x \in S : \omega(x, S) = \tau(s)(x) \wedge \omega(S, x) = \tau(t)(x)$.  For each $x$, $y \in S$, if $\omega (x, y) > 0$ then $(x, y) \in X$ and we can set $f((x, 0), (y, 1)) = \omega (x, y)$ and $k((x, 0), (y, 1)) = \delta_\tau (x, y)$.  This satisfies constraint (a), because $c((x, 0), (y, 1)) = 1$ and $\omega (x, y) \leq \tau(s)(x) \leq 1$.  In order to comply with constraint (b) as well, for each $x \in S$, we set $f(s, (x, 0)) = \tau(s)(x)$ and $f((x, 1), t) = \tau(t)(x)$.  The flow on all other edges is zero.  It follows that $|f| = 1$.  Moreover, since $\omega$ is an optimal coupling, we know that $\omega \cdot \delta_\tau = \delta_\tau (s, t)$.  As a result, $|k_f| = \sum_{e \in E} k(e)f(e) = \delta_\tau (s, t)$.
\end{proof}

\section{Experiments}
\label{appendix:experiments}

Below is a description of jpf-probabilistic's randomized algorithms utilized in our experiments.
\begin{itemize}
    \item Fair Biased Coin: makes a fair coin from a biased coin, where $p$ denotes the probability by which the biased coin tosses heads.  We check the probability that the coin toss results in heads. \cite{vN51}
    \item Majority Element: a Monte Carlo algorithm that determines whether an integer array has a majority element (appears more than half of the time in the array).  The parameter $s$ denotes the size of the given array, $t$ denotes the number of trials, and $m$ denotes the amount of times that the majority element occurs in the array.  We check the probability that the algorithm erroneously reports that the array does not have a majority element. \cite{MR95}
    \item Pollards Integer Factorization: finds a factor of an integer $i$. We check the probability that the algorithm returns $i$, when $i$ is not prime. \cite{P75}
    \item Queens: attempts to place a queen on each row of an $n \times n$ chess board such that no queen can attack another.  We check the probability of success. \cite{B10}
\end{itemize}

%% file: text/discontinuity.tex
\section{Discontinuity}
\label{appendix:discontinuity}

\begin{definition}
For each $a \in L$, the set $S_a$ is defined by
\[
S_a = \{\, s \in S \mid \ell(s) = a \,\}.
\]
\end{definition}

\begin{proposition}
\label{proposition:bisimilar-same-label}
For all $s$, $t \in S$, if $s \sim t$ then $\ell(\support(\tau(s))) = \ell(\support(\tau(t)))$ and for all $a \in L$, we have $\tau(s)(S_a) = \tau(t)(S_a)$.
\end{proposition}
\begin{proof}
Let $s$, $t \in S$.  Assume that $s \sim t$.  Then there exists a bisimulation $R$ with $(s, t) \in R$.  Therefore, there exists $\omega \in \Creal(\tau(s), \tau(t))$ such that $\support(\omega) \subseteq R$.  Since $R$ is a bisimulation, for all $(u, v) \in R$, $\ell(u) = \ell(v)$.  Hence, 
\begin{equation}
\label{equation:bisimulation-same-label}
\forall (u, v) \in \support(\omega) : \ell(u) = \ell(v).
\end{equation}

For all $a \in L$,
\begin{align*}
\tau(s)(S_a) 
& = \omega(S_a, S)
\proofcomment{$\omega \in \Creal(\tau(s), \tau(t))$}\\
& = \omega(S_a, S_a)
\proofcomment{(\ref{equation:bisimulation-same-label})}\\
& = \omega(S, S_a)
\proofcomment{(\ref{equation:bisimulation-same-label})}\\
& = \tau(t)(S_a)
\proofcomment{$\omega \in \Creal(\tau(s), \tau(t))$}
\end{align*}

Assume that $a \in \ell(\support(\tau(s)))$.   Then $\tau(t)(S_a) = \tau(s)(S_a) > 0$.  As a result, $a \in \ell(\support(\tau(t)))$.  Hence, $\ell(\support(\tau(s))) \subseteq \ell(\support(\tau(t)))$.  The other inclusion can be proved by a symmetric argument.
\end{proof}

From Proposition~\ref{proposition:bisimilar-same-label} we can deduce that there exists a function $\mathrm{label}_\tau : S \to L$ such that for all $s$, $t \in S$,
\begin{itemize}
\item 
$\mathrm{label}_\tau(s) \in \ell(\support(\tau(s)))$ and
\item 
if $s \sim t$ then $\mathrm{label}_\tau(s) = \mathrm{label}_\tau(t)$.
\end{itemize}
Since, by assumption, $|\ell(S)| \geq 2$, there exists a function $\overline{\mathrm{label}}_\tau : S \to L$ such that for all $s \in S$,
\begin{itemize}
\item 
$\overline{\mathrm{label}}_\tau(s) \in \ell(S)$,
\item 
$\overline{\mathrm{label}}_\tau(s) \not= \mathrm{label}_\tau(s)$, and
\end{itemize}

There exist functions $\mathrm{minus}_\tau : S \to S$ and $\mathrm{plus}_\tau : S \to S$ such that for all $s \in S$,
\begin{itemize}
\item 
$\mathrm{minus}_\tau(s) \in \support(\tau(s))$,
\item 
$\ell(\mathrm{minus}_\tau(s)) =  \mathrm{label}_\tau(s)$, and
\item 
$\ell(\mathrm{plus}_\tau(s)) =  \overline{\mathrm{label}}_\tau(s)$.
\end{itemize}

Let the function $\mathrm{index} : S \to \{ 1, \ldots, |S| \}$ be onto.  Then $\mathrm{index}$ must also be injective and, hence, bijective.  Note that for all $s \in S$, $\tau(s)(\mathrm{minus}_\tau(s)) > 0$.  We define 
\[
E_\tau = \left (0, \min_{s \in S} \frac{\tau(s)(\mathrm{minus}_\tau(s))}{\mathrm{index}(s)} \right ].
\]

\begin{definition}
For all $\varepsilon \in E_\tau$, the function $\tau_\varepsilon : S \to \Dreal(S)$ is defined by
\[
\tau_\varepsilon(s)(t) = \left \{
\begin{array}{ll}
\tau(s)(t) - \mathrm{index}(s) \, \varepsilon & \hspace{0.5cm} \mbox{if $t = \mathrm{minus}_\tau(s)$}\\
\tau(s)(t) + \mathrm{index}(s) \, \varepsilon & \hspace{0.5cm} \mbox{if $t = \mathrm{plus}_\tau(s)$}\\
\tau(s)(t) & \hspace{0.5cm} \mbox{otherwise.}
\end{array}
\right .
\]
\end{definition}

\begin{proposition}
\label{proposition:tau-affine}
For all $\varepsilon \in E_\tau$, and $s \in S$, there exists $\kappa_s : S \to \mathbb{Z}$ such that $\tau_\varepsilon(s) = \tau(s) + \varepsilon \, \kappa_s$ and $\kappa_s(S) = 0$.
\end{proposition}
\begin{proof}
Let $\varepsilon \in E_\tau$ and $s \in S$.  We define the function $\kappa_s : S \to \mathbb{Z}$ by
\[
\kappa_s(t) = \left \{
\begin{array}{ll}
-\mathrm{index}(s) & \hspace{0.5cm} \mbox{if $t = \mathrm{minus}_\tau(s)$}\\
\mathrm{index}(s) & \hspace{0.5cm} \mbox{if $t = \mathrm{plus}_\tau(s)$}\\
0 & \hspace{0.5cm} \mbox{otherwise.}
\end{array}
\right .
\]
We leave it to the reader to verify that $\kappa_s$ has the desired properties.
\end{proof}


\begin{proposition}
\label{proposition:coupling-label-sets}
For all $\mu$, $\nu \in \Dreal(S)$, $\omega \in \Creal(\mu, \nu)$, and $a \in L$, $\omega(S^2_1) \geq | \mu(S_a) - \nu(S_a) |$.
\end{proposition}
\begin{proof}
Let  $\mu$, $\nu \in \Dreal(S)$, $\omega \in \Creal(\mu, \nu)$, and $a \in L$.  Without loss of generality, assume that $\mu(S_a) \geq \nu(S_a)$.  Then
\begin{align*}
\omega(S_a \times S_a)
& \leq \omega(S \times S_a)\\
& = \nu(S_a)
\proofcomment{$\omega \in \Creal(\mu, \nu)$}
\end{align*}
and
\begin{align*}
\omega  \left (\bigcup_{b \not= a} S_b \times S_b \right )
& \leq \omega  \left (\bigcup_{b \not= a} S_b \times S \right )\\
& = \mu \left (\bigcup_{b \not= a} S_b \right )
\proofcomment{$\omega \in \Creal(\mu, \nu)$}\\
& = 1 - \mu(S_a).
\end{align*}
Therefore,
\begin{align*}
\omega(S^2_1) 
& = \omega \left ((S \times S) \setminus \bigcup_{b \in L} S_b \times S_b \right )\\
& = 1 - \omega(S_a \times S_a) - \omega \left ( \bigcup_{b \not= a} S_b \times S_b \right )\\
& \geq 1 - \nu(S_a) - (1 - \mu(S_a))\\
& = \mu(S_a) - \nu(S_a)\\
& = | \mu(S_a) - \nu(S_a) |
\proofcomment{$\mu(S_a) \geq \nu(S_a)$} \qedhere
\end{align*}
\end{proof}

\begin{corollary}
\label{corollary:bad-pairs}
For all $s$, $t \in S$, $\varepsilon \in E_\tau$, and $\omega \in \Creal(\tau_\varepsilon(s), \tau_\varepsilon(t))$,  if $s \sim t$ and $s \not = t$ then $\omega(S^2_1) \geq \varepsilon$.
\end{corollary}
\begin{proof}
Let $s$, $t \in S$ and $\varepsilon \in E_\tau$.   Assume that $s \sim t$ and $s \not= t$.  Let $a = \mathrm{label}_\tau(s)$.  Then
\begin{align*}
| \tau_\varepsilon(s)(S_a) - \tau_\varepsilon(t)(S_a) |
& = | (\tau(s)(S_a) - \mathrm{index}(s) \varepsilon) - (\tau(t)(S_a) - \mathrm{index}(t) \varepsilon) |\\
& = | \mathrm{index}(t) \varepsilon - \mathrm{index}(s) \varepsilon |
\proofcomment{Proposition~\ref{proposition:bisimilar-same-label}}\\
& \geq \varepsilon
\proofcomment{$\mathrm{index}(s)  \not= \mathrm{index}(t)$ since $s \not= t$}
\end{align*}
Let $\omega \in \Creal(\tau_\varepsilon(s), \tau_\varepsilon(t))$.  From the above and  Proposition~\ref{proposition:coupling-label-sets} we can conclude $\omega(S^2_1) \geq \varepsilon$.
\end{proof}


\begin{proposition}
\label{proposition:affine}
Let $\mu$, $\nu \in \Sreal(X)$ with $\mu(X) = \nu(X)$.  Let $\kappa$, $\lambda : X \to \mathbb{Z}$ with $\kappa(X) = \lambda(X)$.  Let $E \subseteq (0, 1]$ be an infinite set with $0 \in \overline{E}$.  Assume that for all $\varepsilon \in E$, $\mu_\varepsilon = \mu + \varepsilon \, \kappa \in \Sreal(X)$, $\nu_\varepsilon = \nu + \varepsilon \, \lambda \in \Sreal(X)$, $\omega_\varepsilon \in V(\Creal(\mu_\varepsilon, \nu_\varepsilon))$ and all $\omega_\varepsilon$ have the same support graph.  Then there exists $\omega \in V(\Creal(\mu, \nu))$ and $\pi : X \times X \to \mathbb{Z}$ such that for all $\varepsilon \in E$, $\omega_\varepsilon = \omega + \varepsilon \, \pi$.
\end{proposition}
\begin{proof}
Let $\mu$, $\nu \in \Sreal(X)$ with $\mu(X) = \nu(X)$.  Let $\kappa$, $\lambda : X \to \mathbb{Z}$ with $\kappa(X) = \lambda(X)$.  Let $E \subseteq (0, 1]$ be an infinite set with $0 \in \overline{E}$.  Assume that for all $\varepsilon \in E$, $\mu_\varepsilon = \mu + \varepsilon \, \kappa \in \Sreal(X)$, $\nu_\varepsilon = \nu + \varepsilon \, \lambda \in \Sreal(X)$, $\omega_\varepsilon \in V(\Creal(\mu_\varepsilon, \nu_\varepsilon))$ and all $\omega_\varepsilon$ have the same support graph $G$. We prove this proposition by induction on $G$.

In the base case, $G$ is empty.  Hence, $\omega_\varepsilon(x, y) = 0$ for all $x$, $y \in X$ and $\varepsilon \in E$.  Therefore, 
\begin{align*}
0 & = \omega_\varepsilon(x, X)\\
& = \mu_\varepsilon(x)
\proofcomment{$\omega_\varepsilon \in V(\Creal(\mu_\varepsilon, \nu_\varepsilon))$}\\
&  = \mu(x) + \varepsilon \, \kappa(x)
\end{align*}
for all $x \in X$ and $\varepsilon \in E$.  Since $E$ has more than one element, we can conclude that $\mu(x) = 0$ and $\kappa(x) = 0$ for all $x \in X$.  Similarly, we can deduce that $\nu(y) = 0$ and $\lambda(y) = 0$ for all $y \in X$.  Let $\omega(x, y) = 0$ and $\pi(x, y) = 0$ for all $x$, $y \in X$.  Then $\Creal(\mu, \nu) = \{ \omega \}$ and, hence, $\omega \in V(\Creal(\mu, \nu))$, and $\omega_\varepsilon = \omega + \varepsilon \, \pi$ for all $\varepsilon \in E$.

In the inductive case, $G$ is nonempty.  Since $\omega_\varepsilon$ is assumed to be a vertex, $G$ is a forest by Proposition~\ref{proposition:forest}.  Assume, without loss of generality, that $\{ (x, 0), (y, 1) \}$ is an edge of $G$ and that $(x, 0)$ is a leaf.  

\begin{center}
\begin{tikzpicture}[xscale=1.6,yscale=0.8]
\node[ellipse, draw, fit={(0, 1.2) (2, 0) (0, -1.2) (2.2, -1.2)}, fill=ACMBlue!30] {};

\vertex (x) at (0, 1) {$x$};
\vertex (y) at (2, 0) {$y$};

\node at (0, 0.5) {$\mu(x)$};
\node at (2, -0.5) {$\nu(y)$};

\draw[-] (x) to (y);
\end{tikzpicture}
\end{center}

Then 
\begin{align*}
\omega_\varepsilon(x, y) 
& = \omega_\varepsilon(x, X)
\proofcomment{$(x, 0)$ is a leaf of $G$}\\
& = \mu_\varepsilon(x)
\proofcomment{$\omega_\varepsilon \in V(\Creal(\mu_\varepsilon, \nu_\varepsilon))$}
\end{align*}
and 
\begin{align*}
\mu_\varepsilon(x) 
& = \omega_\varepsilon(x, y)\\
&  \leq \omega_\varepsilon(X, y)\\
& = \nu_\varepsilon(y)
\proofcomment{$\omega_\varepsilon \in V(\Creal(\mu_\varepsilon, \nu_\varepsilon))$}
\end{align*}
for all $\varepsilon \in E$.  Hence, $\mu(x) + \varepsilon \, \kappa(x) \leq \nu(y) + \varepsilon \, \lambda(y)$ for all $\varepsilon \in E$.  Therefore, $\mu(x) \leq \nu(y) + \varepsilon \, (\lambda(y) - \kappa(x))$ for all $\varepsilon \in E$.  By assumption, $0 \in \overline{E}$.  Since multiplication (by $\lambda(y) - \kappa(x$)) is continuous and the continuous image of the closure of a set is a subset of the closure of the continuous image of the set (see, for example, \cite[Proposition~1.4.1(v)]{E89}), we can conclude that $0 \in \overline{\{\, \varepsilon \, (\lambda(y) - \kappa(x)) \mid \varepsilon \in E \,\}}$.  Hence, $\mu(x) \leq \nu(y)$.  Let
\begin{align*}
\mu'(z) & = \left \{
\begin{array}{ll}
0 & \hspace{0.5cm} \mbox{if $z = x$}\\
\mu(z) & \hspace{0.5cm} \mbox{otherwise}
\end{array}
\right .\\
\nu'(z) & = \left \{
\begin{array}{ll}
\nu(y) - \mu(x) & \hspace{0.5cm} \mbox{if $z = y$}\\
\nu(z) & \hspace{0.5cm} \mbox{otherwise}
\end{array}
\right .\\
\kappa'(z) & = \left \{
\begin{array}{ll}
0 & \hspace{0.5cm} \mbox{if $z = x$}\\
\kappa(z) & \hspace{0.5cm} \mbox{otherwise}
\end{array}
\right .\\
\lambda'(z) & = \left \{
\begin{array}{ll}
\lambda(y) - \kappa(x) & \hspace{0.5cm} \mbox{if $z = y$}\\
\lambda(z) & \hspace{0.5cm} \mbox{otherwise}
\end{array}
\right .
\end{align*}
Since $\mu(x) \leq \nu(y)$, we can conclude that $\nu' \in \Sreal(X)$.  Obviously, $\mu' \in \Sreal(X)$. We have that 
\begin{align*}
\mu'(X) 
& = \mu(X) - \mu(x)\\
& = \nu(X) - \mu(x)\\
& = \nu'(X).
\end{align*}
Furthermore, 
\begin{align*}
\kappa'(X) 
& = \kappa(X) - \kappa(x)\\
& = \lambda(X) - \kappa(x)\\
& = \lambda'(X).  
\end{align*}

\begin{center}
\begin{tikzpicture}[xscale=1.6,yscale=0.8]
\node[ellipse, draw, fit={(0, 1.2) (2, 0) (0, -1.2) (2.2, -1.2)}, fill=ACMBlue!30] {};

\vertex (x) at (0, 1) {$x$};
\vertex (y) at (2, 0) {$y$};

\node at (0, 0.5) {$0$};
\node at (2, -0.5) {$\nu(y) - \mu(x)$};

\draw[-] (x) to (y);
\end{tikzpicture}
\end{center}

For each $\varepsilon \in E$, let
\begin{align*}
\mu_\varepsilon'(z) & = \left \{
\begin{array}{ll}
0 & \hspace{0.5cm} \mbox{if $z = x$}\\
\mu_\varepsilon(z) & \hspace{0.5cm} \mbox{otherwise}
\end{array}
\right .\\
\nu_\varepsilon'(z) & = \left \{
\begin{array}{ll}
\nu_\varepsilon(y) - \mu_\varepsilon(x) & \hspace{0.5cm} \mbox{if $z = y$}\\
\nu_\varepsilon(z) & \hspace{0.5cm} \mbox{otherwise}
\end{array}
\right .\\
\omega_\varepsilon'(u, v) & = \left \{
\begin{array}{ll}
0 & \hspace{0.5cm} \mbox{if $(u, v) = (x, y)$}\\
\omega_\varepsilon(u, v) & \hspace{0.5cm} \mbox{otherwise}
\end{array}
\right .
\end{align*}
Since $\mu_\varepsilon(x) \leq \nu_\varepsilon(y)$, we can conclude that $\nu_\varepsilon' \in \Sreal(X)$.  Obviously, $\mu_\varepsilon' \in \Sreal(X)$.  
Since 
\begin{align*}
\mu_\varepsilon'(x) 
& = 0\\
& = \mu'(x) + \varepsilon \, \kappa'(x)
\end{align*}
and 
\begin{align*}
\mu_\varepsilon'(z) 
& = \mu_\varepsilon(z)\\
& = \mu(z) + \varepsilon \, \kappa(z)\\
& = \mu'(z) + \varepsilon \, \kappa'(z)
\end{align*}
for all $z \not= x$, we have that $\mu_\varepsilon' = \mu' + \varepsilon \, \kappa' $.  Because 
\begin{align*}
\nu_\varepsilon'(y) 
& = \nu_\varepsilon(y) - \mu_\varepsilon(x)\\
& = (\nu(y) + \varepsilon \, \lambda(y)) - (\mu(x) + \varepsilon \, \kappa(x))\\
& = (\nu(y) - \mu(x)) + \varepsilon\, (\lambda(y) - \kappa(x))\\
&  = \nu'(y) + \varepsilon \, \lambda'(y)
\end{align*}
and 
\begin{align*}
\nu_\varepsilon'(z) 
& = \nu_\varepsilon(z)\\
& = \nu(z) + \varepsilon \, \lambda(z)\\
&  = \nu'(z) + \varepsilon \, \lambda'(z)
\end{align*}
for all $z \not= y$, we have that $\nu_\varepsilon' = \nu' + \varepsilon \, \lambda'$.  

Since we have that
\begin{align*}
\omega_\varepsilon'(x, X)
& = \omega_\varepsilon(x, X) - \mu_\varepsilon(x)
\proofcomment{$\omega_\varepsilon(x, y) = \mu_\varepsilon(x)$ and $\omega_\varepsilon'(x, y) = 0$}\\
& = \mu_\varepsilon(x) - \mu_\varepsilon(x)
\proofcomment{$\omega_\varepsilon \in \Creal(\mu_\varepsilon, \nu_\varepsilon)$}\\
& = 0\\
& = \mu_\varepsilon'(x)
\end{align*}
and for all $z \not= x$
\begin{align*}
\omega_\varepsilon'(z, X)
& = \omega_\varepsilon(z, X)\\
& = \mu_\varepsilon(z)
\proofcomment{$\omega_\varepsilon \in \Creal(\mu_\varepsilon, \nu_\varepsilon)$}\\
& = \mu_\varepsilon'(z)
\end{align*}
and
\begin{align*}
\omega_\varepsilon'(X, y)
& = \omega_\varepsilon(X, y) - \mu_\varepsilon(x)
\proofcomment{$\omega_\varepsilon(x, y) = \mu_\varepsilon(x)$ and $\omega_\varepsilon'(x, y) = 0$}\\
& = \nu_\varepsilon(y) - \mu_\varepsilon(x)
\proofcomment{$\omega_\varepsilon \in \Creal(\mu_\varepsilon, \nu_\varepsilon)$}\\
& = \nu_\varepsilon'(y)
\end{align*}
and for all $z \not= y$
\begin{align*}
\omega_\varepsilon'(X, z)
& = \omega_\varepsilon(X, z)\\
& = \nu_\varepsilon(z)
\proofcomment{$\omega_\varepsilon \in \Creal(\mu_\varepsilon, \nu_\varepsilon)$}\\
& = \nu_\varepsilon'(z)
\end{align*}
we can deduce that $\omega_\varepsilon' \in \Creal(\mu_\varepsilon', \nu_\varepsilon')$.  

\begin{center}
\begin{tikzpicture}[xscale=1.6,yscale=0.8]

\node[ellipse, draw, fit={(0, 1.2) (2, 0) (0, -1.2) (2.2, -1.2)}, fill=ACMBlue!30] {};
\node[ellipse, draw, fit={(2, 0.2) (0.2, -1) (1.4, -1)}, fill=ACMOrange!40] {};

\vertex (x) at (0, 1) {$x$};
\vertex (y) at (2, 0) {$y$};

\draw[-] (x) to (y);
\end{tikzpicture}
\end{center}

Recall that all $\omega_\varepsilon$ have support graph $G$.  Let $G' = G \setminus \{ \{ (x, 0), (y, 1) \} \}$.  From the definition of $\omega_\varepsilon'$ we can conclude that all $\omega_\varepsilon'$ have support graph $G'$.  Since $G'$ is a subgraph of $G$, which is a forest, the former is also a forest.  Therefore, by Proposition~\ref{proposition:forest}, $\omega_\varepsilon'$ is a vertex.  By induction, there exist $\omega' \in V(\Creal(\mu', \nu'))$ and $\pi' : X \times X \to \mathbb{Z}$ such that for all $\varepsilon \in E$, $\omega_\varepsilon' = \omega' + \varepsilon \, \pi'$.
Let
\begin{align*}
\omega(u, v) & = \left \{
\begin{array}{ll}
\mu(x) & \hspace{0.5cm} \mbox{if $(u, v) = (x, y)$}\\
\omega'(u, v) & \hspace{0.5cm} \mbox{otherwise}
\end{array}
\right .\\
\pi(u, v) & = \left \{
\begin{array}{ll}
\kappa(x) & \hspace{0.5cm} \mbox{if $(u, v) = (x, y)$}\\
\pi'(u, v) & \hspace{0.5cm} \mbox{otherwise}
\end{array}
\right .
\end{align*}

Note that
\begin{align*}
\omega'(x, y) 
& \leq  \omega'(x, X)\\
& = \mu'(x)
\proofcomment{$\omega' \in V(\Creal(\mu', \nu'))$}\\
& = 0.
\end{align*}
Hence, $\omega'(x, y) = 0$.  Since
\begin{align*}
\omega(x, X)
& =  \omega'(x, X) - \omega'(x, y) + \mu(x)\\
& = \mu'(x) + \mu(x)
\proofcomment{$\omega' \in V(\Creal(\mu', \nu'))$ and $\omega'(x, y) = 0$}\\
& = \mu(x)
\end{align*}
and for all $z \not= x$,
\begin{align*}
\omega(z, X)
& = \omega'(z, X)\\
& = \mu'(z)
\proofcomment{$\omega' \in V(\Creal(\mu', \nu'))$}\\
& = \mu(z)
\end{align*}
and
\begin{align*}
\omega(X, y)
& = \omega'(X, y) - \omega'(x, y) + \mu(x)\\
& = \nu'(y) + \mu(x)
\proofcomment{$\omega' \in V(\Creal(\mu', \nu'))$ and $\omega'(x, y) = 0$}\\
& = \nu(y)
\end{align*}
and for all $z \not= y$,
\begin{align*}
\omega(X, z)
& = \omega'(X, z)\\
& = \nu'(z)
\proofcomment{$\omega' \in V(\Creal(\mu', \nu'))$}\\
& = \nu(z)
\end{align*}
we can conclude that $\omega \in \Creal(\mu, \nu)$.  Since the support graph of $\omega$ is $G$ and $G$ is a forest, we can conclude from Proposition~\ref{proposition:forest} that $\omega$ is a vertex.

Since
\begin{align*}
\omega_\varepsilon(x, y)
& = \mu_\varepsilon(x)\\
& = \mu(x) + \varepsilon \, \kappa(x)\\
& = \omega(x, y) + \varepsilon \, \pi(x, y)
\end{align*}
and for all $(u, v) \not= (x, y)$,
\begin{align*}
\omega_\varepsilon(u, v)
& = \omega_\varepsilon'(u, v)\\
& = \omega'(u, v) + \varepsilon \, \pi'(u, v)\\
& = \omega(u, v) + \varepsilon \, \pi(u, v)
\end{align*}
we can conclude that $\omega_\varepsilon = \omega + \varepsilon \, \pi$.
\end{proof}

%
%


\begin{proposition}
\label{proposition:optimal-vertex}
There exists $P \in \mathcal{P}$ such that $P$ is an optimal vertex policy.
\end{proposition}
\begin{proof}
We first define a vertex policy $P \in \mathcal{P}$ and then prove that it is optimal.  Let $(s, t) \in S^2_1$.  By definition, $\support(P(s, t)) = \{ (s, t) \}$ and, hence, $P(s, t)$ is defined by
\[
P(s, t)(u, v) = \left \{
\begin{array}{ll}
1 & \hspace{0.5cm} \mbox{if $(u, v) = (s, t)$}\\
0 & \hspace{0.5cm} \mbox{otherwise.}
\end{array}
\right .
\]
Let $(s, t) \in S^2_\Delta \cup S^2_{0?}$.  We have that
\[
\delta_\tau(s, t) 
= \Delta_\tau(\delta_\tau)(s, t)
= \inf_{\omega \in \Creal(\tau(s), \tau(t))} \omega \cdot \delta_\tau.
\]
The function mapping $\omega$ to $\omega \cdot \delta_\tau$ is concave.  Since $\Omega(\tau(s), \tau(t))$ is a convex polytope and a concave function on a convex polytope attains its minimum at a vertex of the polytope (see, for example, \cite[page~260]{KW67}), we can conclude that there exists $\omega_{st} \in V(\Omega(\tau(s), \tau(t)))$ such that 
\[
\delta_\tau(s, t) = \omega_{st} \cdot \delta_\tau,
\]
In this case, $P(s, t) = \omega_{st}$.

It remains to show that the above defined vertex policy is optimal.  According to Proposition~\ref{proposition:optimal-characterization}, it suffices to show that for all $s$, $t \in S$, $\Gamma_P(\delta_\tau)(s, t) \leq \delta_\tau(s, t)$.  We distinguish the following cases.
\begin{itemize}
\item 
If $(s, t) \in S^2_1$  then
\[
\Gamma_P(\delta_\tau)(s, t)
= 1
= \Delta_\tau(\delta_\tau)(s, t)
= \delta_\tau(s, t).
\]
\item 
If $(s, t) \in S^2_\Delta \cup S^2_{0?}$ then
\[
\Gamma_P(\delta_\tau)(s, t)
= P(s, t) \cdot \delta_\tau
= \omega_{st} \cdot \delta_\tau
= \delta_\tau(s, t). \qedhere
\] 
\end{itemize}
\end{proof}

\begin{proposition}
\label{proposition:limit-vertex-policy}
There exist a sequence $(\varepsilon_n)_n$ in $E_\tau$ converging to 0 and a sequence $(P_n)_n$, with $P_n \in \mathcal{P}_{\tau_{\varepsilon_n}}$, of optimal vertex policies,
a vertex policy $P \in \mathcal{P}_\tau$, and $D : (S \times S) \to (S \times S) \to \mathbb{Z}$ such that $P_n = P + \varepsilon_n \, D$.
\end{proposition}
\begin{proof}
Let $\varepsilon \in E_\tau$ and $(s, t) \in S^2_\Delta \cup S^2_{0?}$.  According to Proposition~\ref{proposition:tau-affine}, there exist $\kappa_s : S \to \mathbb{Z}$ and $\lambda_t : S \to \mathbb{Z}$ such that $\tau_\varepsilon(s) = \tau(s) + \varepsilon \, \kappa_s$ and $\tau_\varepsilon(t) = \tau(t) + \varepsilon \, \lambda_t$ and $\kappa_s(S) = 0 = \lambda_t(S)$.  By Proposition~\ref{proposition:optimal-vertex}, there exists an optimal vertex policy $P_{\varepsilon} \in \mathcal{P}_{\tau_\varepsilon}$.  Therefore, $P_{\varepsilon}(s, t) \in V(\Creal(\tau_\varepsilon(s), \tau_\varepsilon(t)))$.

Since the set $E_\tau$ is infinite and $0 \in \overline{E_\tau}$, there exists a sequence $(\varepsilon_n)_n$ in $E_\tau$ with all elements distinct that converges to 0.  Because the set $S$ is finite, there are only finitely many support graphs for $\omega \in \Dreal(S \times S)$.   Hence, there exists a subsequence $(\varepsilon_{f(n)})_n$ of $(\varepsilon_n)_n$, which obviously also converges to 0, such that for all $(s, t) \in S^2_\Delta \cup S^2_{0?}$, $\{\, P_{\varepsilon_{f(n)}}(s, t) \mid n \in \nat \,\}$ all have the same support graph.  Hence, there exists an infinite subset $E = \{\, \varepsilon_{f(n)} \mid n \in \nat \,\}$ of $E_\tau$ such that for all $(s, t) \in S^2_\Delta \cup S^2_{0?}$, $\{\, P_{\varepsilon}(s, t) \mid \varepsilon \in E \,\}$ have the same support graph and $0 \in \overline{E}$.

Let $(s, t) \in S^2_\Delta \cup S^2_{0?}$.  Now we are in a position to apply Proposition~\ref{proposition:affine}, where $\mu = \tau(s)$, $\nu = \tau(t)$, $\kappa = \kappa_s$, $\lambda = \lambda_t$, $\mu_\varepsilon = \tau_\varepsilon(s)$, $\nu_\varepsilon = \tau_\varepsilon(t)$, and $\omega_\varepsilon = P_\varepsilon(s, t)$.  From Proposition~\ref{proposition:affine} we can conclude that there exist $\omega_{st} \in V(\Creal(\tau(s), \tau(t)))$ and $\pi_{st} : S \times S \to \mathbb{Z}$ such that 
for all $\varepsilon \in E$,
\begin{equation}
\label{equation:affine}
P_\varepsilon(s, t) = \omega_{st} + \varepsilon \, \pi_{st}.
\end{equation}

We define the policy $P \in \mathcal{P}_\tau$ by
\[
P(s, t)(u, v) = \left \{
\begin{array}{ll}
\omega_{st}(u, v)
& \hspace{0.5cm} \mbox{if $(s, t) \in S^2_\Delta \cup S^2_{0?}$}\\
1
& \hspace{0.5cm} \mbox{if $(s, t) \in S^2_1$ and $(u, v) = (s, t)$.}
\end{array}
\right .
\]
and the function $D : (S \times S) \to (S \times S) \to \mathbb{Z}$ by
\[
D(s, t)(u, v) = \left \{
\begin{array}{ll}
\pi_{st}(u, v)
& \hspace{0.5cm} \mbox{if $(s, t) \in S^2_\Delta \cup S^2_{0?}$}\\
0
& \hspace{0.5cm} \mbox{otherwise.}
\end{array}
\right .
\]
Since for all $(s, t) \in S^2_\Delta \cup S^2_{0?}$, $\omega_{st} \in V(\Creal(\tau(s), \tau(t)))$, $P$ is a vertex policy.  From (\ref{equation:affine}) and the definitions of $P$ and $D$ it immediately follows that
$P_{\varepsilon_{f(n)}} = P + \varepsilon_{f(n)}\, D$.
In conclusion, the sequence $(\varepsilon_{f(n)})_n$ in $E_\tau$ converges to 0, $P_{\varepsilon_{f(n)}}$ is an optimal vertex policy for each $n \in \nat$, and $P$ is a vertex policy such that $P_{\varepsilon_{f(n)}} = P + \varepsilon_{f(n)}\, D$ for all $n \in \nat$.
\end{proof}

As the following example shows, the limit of a converging sequence of optimal vertex policies need not be optimal.

\begin{example}
Consider the following family of labelled Markov chains for $\varepsilon \in [0, \frac{1}{2}]$.
\begin{center}
\begin{tikzpicture}[xscale=0.96,yscale=1.4]
\node[regular polygon, regular polygon sides=6, fill=ACMYellow!40, draw=black, inner sep=0pt, minimum size=18pt] (a1) at (2.5,2) {$a_1$};
\node[regular polygon, regular polygon sides=6, fill=ACMYellow!40, draw=black, inner sep=0pt, minimum size=18pt] (a2) at (7,2) {$a_2$};

\node[regular polygon, regular polygon sides=6, fill=ACMYellow!40, draw=black, inner sep=0pt, minimum size=18pt] (b1) at (1,1) {$b_1$};
\node[regular polygon, regular polygon sides=6, fill=ACMYellow!40, draw=black, inner sep=0pt, minimum size=18pt] (b2) at (4,1) {$b_2$};
\node[regular polygon, regular polygon sides=6, fill=ACMYellow!40, draw=black, inner sep=0pt, minimum size=18pt] (b3) at (6,1) {$b_3$};
\node[rectangle, fill=ACMPurple!25, draw=black, inner sep=0pt, minimum size=15pt] (c) at (8,1) {$c$};

\node[regular polygon, regular polygon sides=6, fill=ACMYellow!40, draw=black, inner sep=0pt, minimum size=18pt] (d1) at (0,0) {$d_1$};
\node[rectangle, fill=ACMPurple!25, draw=black, inner sep=0pt, minimum size=15pt] (e1) at (2, 0) {$e_1$};
\node[rectangle, fill=ACMPurple!25, draw=black, inner sep=0pt, minimum size=15pt] (f) at (3,0) {$f$};
\node[regular polygon, regular polygon sides=6, fill=ACMYellow!40, draw=black, inner sep=0pt, minimum size=18pt] (d2) at (5,0) {$d_2$};]
\node[rectangle, fill=ACMPurple!25, draw=black, inner sep=0pt, minimum size=15pt] (e2) at (7, 0) {$e_2$};

\draw[->] (a1) to node[midway,above left] {$\frac{1}{2}$} (b1);
\draw[->] (a1) to node[midway,above right] {$\frac{1}{2}$} (b2);
\draw[->] (a2) to node[midway,above left] {$\frac{1}{2}$} (b3);
\draw[->] (a2) to node[midway,above right] {$\frac{1}{2}$} (c);
\draw[->] (b1) to node[midway,above left] {$\frac{1}{2}$} (d1);
\draw[->] (b1) to node[midway,above right] {$\frac{1}{2}$} (e1);
\draw[->] (b2) to node[midway,above left] {$\frac{1}{2}$} (f);
\draw[->] (b2) to node[midway,above right] {$\frac{1}{2}$} (d2);
\draw[->] (b3) to node[midway,above left] {$\frac{1}{2}$} (d2);
\draw[->] (b3) to node[midway,above right] {$\frac{1}{2}$} (e2);

\draw[->, loop below] (c) to node {$1$} (c);
\draw[->, loop below] (f) to node {$1$} (f);
\draw[->, loop below] (d1) to node {$\frac{1}{2} + \varepsilon$} (d1);
\draw[->, loop below] (d2) to node {$\frac{1}{2}$} (d2);

\draw[->, bend left=20] (d1) to node[midway,above] {$\frac{1}{2} - \varepsilon$} (e1);
\draw[->, bend left=20] (d2) to node[midway,above] {$\frac{1}{2}$} (e2);
\draw[->, bend left=20] (e1) to node[midway,below] {$1$} (d1);
\draw[->, bend left=20] (e2) to node[midway,below] {$1$} (d2);
\end{tikzpicture}
\end{center}
We denote the transition probability function for $\varepsilon = \frac{1}{n}$ by $\tau_n$ for each $n \in \nat$ and for $\varepsilon = 0$ by $\tau$.  Note that $(\tau_n)_n$ converges to $\tau$.

The equivalence classes under bisimilarity for $\tau$ are shown below.
\begin{center}
\begin{tikzpicture}[scale=0.9]
\node[class, minimum width=1cm] (1) at (1.3,0) {$a_1$};
\node[class] (2) at (3,0) {$a_2$, $b_2$};
\node[class, minimum width=3cm] (1) at (5.8,0) {$b_1$, $b_3$, $d_1$, $d_2$};
\node[class] (3) at (8.6,0) {$c$, $f$};
\node[class] (3) at (10.7,0) {$e_1$, $e_2$};
\end{tikzpicture}
\end{center}
Hence, $\delta_{\tau}(b_1, b_3) = 0$, $\delta_{\tau}(d_1, d_2) = 0$, and  $\delta_{\tau}(e_1, e_2) = 0$.  Since, $\delta_{\tau}(f, e_2) = 1$, it follows that $\delta_{\tau}(b_2, b_3) = \frac{1}{2}$.  Note that, for all $n \in \nat$, $\delta_{\tau_n}(e_1, e_2) = \delta_{\tau_n}(d_1, d_2)$.  Then $\delta_{\tau_n}(d_1, d_2) = (1 - \varepsilon) \delta_{\tau_n}(d_1, d_2) + \varepsilon = 1$.  As a consequence, for all $n \in \nat$, $\delta_{\tau_n}(e_1, e_2) = 1$, $\delta_{\tau_n}(b_1, b_3) = 1$, $\delta_{\tau_n}(b_2, b_3) = \frac{1}{2}$.

Let $P_n$ be optimal vertex policies for $\tau_n$ for all $n \in \nat$.  Since $\delta_{\tau_n}(b_1, b_3) = 1$ and $\delta_{\tau}(b_2, b_3) = \frac{1}{2}$, $P_n(a_1, a_2)$ can be depicted as follows.
\begin{center}
\begin{tikzpicture}
\node[smallstate, fill=ACMYellow!40] (a1) at (0,1) {$a_1$}; 
\node[smallstate, fill=ACMYellow!40] (b1) at (1,2) {$b_1$}; 
\node[smallstate, fill=ACMYellow!40] (b2) at (1,0) {$b_2$}; 

\node[smallstate, fill=ACMYellow!40] (a2) at (4,1) {$a_2$}; 
\node[smallstate, fill=ACMYellow!40] (b3) at (3,2) {$b_3$}; 
\node[smallstate, fill=ACMPurple!25] (c) at (3,0) {$c$}; 

\draw[->] (a1) to node[midway,above left] {$\frac{1}{2}$} (b1);
\draw[->] (a1) to node[midway,below left] {$\frac{1}{2}$} (b2);
\draw[->] (a2) to node[midway,above right] {$\frac{1}{2}$} (b3);
\draw[->] (a2) to node[midway,below right] {$\frac{1}{2}$} (c);

\draw[-] (b2) to node[pos=0.3,below] {$\frac{1}{2}$} (b3);
\draw[-] (b1) to node[pos=0.3,above] {$\frac{1}{2}$} (c);
\end{tikzpicture}
\end{center}
Let $P$ be any optimal vertex policy for $\tau$.  Since $\delta_{\tau}(b_1, b_3) = 0$ and  $\delta_{\tau}(b_2, b_3) = \frac{1}{2}$, $P(a_1, a_2)$ can be depicted as follows.
\begin{center}
\begin{tikzpicture}
\node[smallstate, fill=ACMYellow!40] (a1) at (0,1) {$a_1$}; 
\node[smallstate, fill=ACMYellow!40] (b1) at (1,2) {$b_1$}; 
\node[smallstate, fill=ACMYellow!40] (b2) at (1,0) {$b_2$}; 

\node[smallstate, fill=ACMYellow!40] (a2) at (4,1) {$a_2$}; 
\node[smallstate, fill=ACMYellow!40] (b3) at (3,2) {$b_3$}; 
\node[smallstate, fill=ACMPurple!25] (c) at (3,0) {$c$}; 

\draw[->] (a1) to node[midway,above left] {$\frac{1}{2}$} (b1);
\draw[->] (a1) to node[midway,below left] {$\frac{1}{2}$} (b2);
\draw[->] (a2) to node[midway,above right] {$\frac{1}{2}$} (b3);
\draw[->] (a2) to node[midway,below right] {$\frac{1}{2}$} (c);

\draw[-] (b2) to node[midway,below] {$\frac{1}{2}$} (c);
\draw[-] (b1) to node[midway,above] {$\frac{1}{2}$} (b3);
\end{tikzpicture}
\end{center}
Note that $(P_n(a_1, a_2))_n$ does not converge to $P(a_1, a_2)$ and in both cases there is a unique optimal policy.\lipicsEnd
\end{example}

\begin{proposition}
\label{proposition:ccc}
For all $P \in \mathcal{P}$, if $C$ is a closed communication class of $\chain{P}$ then
\begin{enumerate}
\item 
$C = \{ (s, t) \}$ for some $(s, t) \in S^2_1$, or
\item 
$C \subseteq S^2_\Delta \cup S^2_{0,\tau}$.
\end{enumerate}
\end{proposition}
\begin{proof}
Let $P \in \mathcal{P}$ and $C$ be a closed communication class of $\chain{P}$.  Let $s$, $t \in S$ and $(s, t) \in C$.  We distinguish the following cases.
\begin{enumerate}[a.]
    \item Suppose $(s, t) \in S^2_1$.  Then, it follows immediately from the definition of $\mathcal{P}$ that $C = \{ (s, t) \}$.
    \item Suppose $(s, t) \in S \times S \setminus S^2_1$.  Let $(u, v) \in C$.  By a., $(u, v) \in S \times S \setminus S^2_1$.  Hence $C \cap S^2_1 = \varnothing$.  By Proposition~\ref{proposition:reach-s1}, we have $\delta_\tau(s, t) = 0$.  Therefore, $(s, t) \in S^2_\Delta \cup S^2_{0,\tau}$. Thus, $C \subseteq S^2_\Delta \cup S^2_{0,\tau}$. \qedhere
\end{enumerate}
\end{proof}

\begin{proposition}
\label{proposition:discontinuity-gamma}
Let $(\varepsilon_n)_n$ be a sequence in $E_\tau$ converging to 0, $(P_n)_n$, with $P_n \in \mathcal{P}_{\tau_{\varepsilon_n}}$, a sequence of optimal vertex policies, a vertex policy $P \in \mathcal{P}_\tau$, and $D : (S \times S) \to (S \times S) \to \mathbb{Z}$ such that $P_n = P + \varepsilon_n D$ for all $n \in \nat$.  Let $C$ be a closed communication class of $\chain{P}$ with $C \subseteq S^2_{0,\tau}$.  For all $(s, t) \in C$, $\lim \inf_n \reachone{P_n}(s, t) > 0$.
\end{proposition}
\begin{proof}
Let $(\varepsilon_n)_n$ be a sequence in $E_\tau$ converging to 0, $(P_n)_n$, with $P_n \in \mathcal{P}_{\tau_{\varepsilon_n}}$, a sequence of optimal vertex policies, a vertex policy $P \in \mathcal{P}_\tau$, and $D : (S \times S) \to (S \times S) \to \mathbb{Z}$ such that $P_n = P + \varepsilon_n D$ for all $n \in \nat$.  Let $C$ be a closed communication class of $\chain{P}$ with $C \subseteq S^2_{0,\tau}$.  Let $(s, t) \in C$.  Let $M = \max_{(s, t) \in C} D(s, t)((S \times S) \setminus C)$.  For all $n \in \nat$, 
\begin{align}
P_n(s, t)(C)
& = 1 - P_n(s, t)((S \times S) \setminus C) \nonumber\\
& = 1 - (P(s, t)((S \times S) \setminus C) + \varepsilon_n D(s, t)((S \times S) \setminus C)) \proofcomment{$P_n = P + \varepsilon_n D$} \nonumber\\
& = 1 - \varepsilon_n D(s, t)((S \times S) \setminus C) \proofcomment{$(s, t) \in C$ and $C$ is a ccc of $\chain{P}$} \nonumber\\
& \geq 1 - \varepsilon_n M
\label{equation:stay-in-C}
\end{align}

\begin{center}
\begin{tikzpicture}[xscale=1.8,yscale=0.9]
\node[ellipse, draw, fit={(0.5, 0.5) (1, 1) (1, 0) (1.5,0.5)}, fill=ACMBlue!30] (C) {};
\node[state] (st) at (1, 0.5) {$s, t$};
\node at (0, 0.5) {$C$};

\node[ellipse, draw, fit={(3,0) (3.5,-1)}] (S21) {};
\node at (3.25,-0.5) {$S^2_1$};

\draw[->,bend right] (st) to node[midway, below] {$\geq \varepsilon_n$} (S21) ;
\draw[->] (st) to[out=10, in=-30] (1.5,2) to[out=150, in=80] (C);
\node at (2.2,2.2) {$\geq 1 - \varepsilon_n M$};
\end{tikzpicture}
\end{center}
Furthermore,
\begin{align*}
& D(s, t)((S \times S) \setminus C)\\
& = \frac{P_0(s, t)((S \times S) \setminus C) - P(s, t)((S \times S) \setminus C)}{\varepsilon_0} \proofcomment{$P_0 = P + \varepsilon_0 D$}\\
& = \frac{P_0(s, t)((S \times S) \setminus C)}{\varepsilon_0} \proofcomment{$(s, t) \in C$ and $C \subseteq S^2_{0,\tau}$ is a ccc of $\chain{P}$}\\
& \geq \frac{P_0(s, t)(S^2_1)}{\varepsilon_0}
\proofcomment{$S^2_1 \subseteq (S \times S) \setminus C$}\\
& \geq 1
\proofcomment{$P_0(s, t)(S^2_1) \geq \varepsilon_0$ by Corollary~\ref{corollary:bad-pairs}}
\end{align*}
Therefore, we can deduce that $M \geq 1$.   For each $n \in \nat$, let
\begin{equation}
\label{equation:minimal-value}
(s_n, t_n) \in \argmin_{(s, t) \in C} \reachone{P_n}(s, t).
\end{equation}
Let $m_n = \reachone{P_n}(s_n, t_n)$.  We have that
\begin{align*}
m_n
& = \reachone{P_n}(s_n, t_n)\\
& = \Gamma_{P_{n}}(\reachone{P_n})(s_n, t_n)\\
& = P_{n}(s_n, t_n) \cdot \reachone{P_n}\\
& = \sum_{(u, v) \in S^2_1} P_{n}(s_n, t_n)(u, v) \, \reachone{P_n}(u, v) + \sum_{(u, v) \in C} P_{n}(s_n, t_n)(u, v) \, \reachone{P_n}(u, v)\\
& \qquad + \sum_{(u, v) \in (S \times S) \setminus (C \cup S^2_1)} P_{n}(s_n, t_n)(u, v) \, \reachone{P_n}(u, v)\\
& \geq \sum_{(u, v) \in S^2_1} P_{n}(s_n, t_n)(u, v) \, \reachone{P_n}(u, v) + \sum_{(u, v) \in C} P_{n}(s_n, t_n)(u, v) \, \reachone{P_n}(u, v)\\
& = P_{n}(s_n, t_n)(S^2_1)
+ \sum_{(u, v) \in C} P_{n}(s_n, t_n)(u, v) \, \reachone{P_n}(u, v)\\
& \quad \proofcomment{$\reachone{P_n}(u, v) =  \Gamma_{P_{n}}(\reachone{P_n})(u, v) = 1$ for all $(u, v) \in S^2_1$}\\
& \geq \varepsilon_{n} + \sum_{(u, v) \in C} P_{n}(s_n, t_n)(u, v) \, \reachone{P_n}(u, v) \proofcomment{Corollary~\ref{corollary:bad-pairs}}\\
& \geq \varepsilon_{n} +  \sum_{(u, v) \in C} P_{n}(s_n, t_n)(u, v) \, m_n \proofcomment{$\reachone{P_n}(u, v) \geq m_n$ for all $(u, v) \in C$}\\
& = \varepsilon_{n} +  P_{n}(s_n, t_n)(C) \, m_n\\
& \geq  \varepsilon_{n} + (1 - \varepsilon_{n} M) \, m_n
\proofcomment{(\ref{equation:stay-in-C})}
\end{align*} 
Since $m_n \geq \varepsilon_{n} + (1 - \varepsilon_{n}  M) \, m_n$ and $0 < \varepsilon_{n}$, we can conclude that $m_n \geq \frac{1}{M}$.  Then
\begin{align*}
\liminf_n \reachone{P_n}(s, t)
& \geq \liminf_n \reachone{P_n}(s_n, t_n)
\proofcomment{(\ref{equation:minimal-value})}\\
& = \liminf_n m_n
\proofcomment{definition of $m_n$}\\
& \geq \frac{1}{M}
\proofcomment{for all $n$, $m_n \geq \frac{1}{M}$} \qedhere
\end{align*}
\end{proof}

\begin{lemma}
\label{lemma:discontinuous}
Let $s$, $t \in S$.  If for all vertex policies $P \in \mathcal{P}_\tau$ that are optimal,
\begin{equation}
\label{equation:reach-delta-discontinuous}
(s, t) \mbox{ reaches } S^2_\Delta \mbox{ with probability } < 1 - \delta_\tau(s, t) \mbox{ in } \chain{P}
\end{equation}
then the function $\delta_{\_}(s, t) : (S \to \Dreal(S)) \to [0, 1]$ is discontinuous at $\tau$. 
\end{lemma}
\begin{proof}
By Proposition~\ref{proposition:limit-vertex-policy}, there exist a sequence $(\varepsilon_n)_n$ in $E_\tau$ converging to 0 and a sequence $(P_n)_n$, with $P_n \in \mathcal{P}_{\tau_{\varepsilon_n}}$, of optimal vertex policies, a vertex policy $P \in \mathcal{P}_\tau$, and $D : (S \times S) \to (S \times S) \to \mathbb{Z}$ such that $P_n = P + \varepsilon_n \, D$.

Since the set $S$ is finite, we can conclude from Proposition~\ref{proposition:unit-interval-compact} and \ref{proposition:finite-product-compact} that $S \times S \to [0, 1]$ is compact.  As a consequence, the sequence $(\reachone{P_n})_n$ has a converging subsequence $(\reachone{P_{f(n)}})_n$.  We have that
\begin{align}
\Gamma_P(\lim_n \reachone{P_{f(n)}})
& = \Gamma_{\lim_n P_{f(n)}}(\lim_n \reachone{P_{f(n)}}) \nonumber\\
& = \lim_n \Gamma_{P_{f(n)}}(\reachone{P_{f(n)}})
\proofcomment{Corollary~\ref{corollary:Gamma-continuous}} \nonumber\\
& = \lim_n \reachone{P_{f(n)}}
\label{equation:lim-gamma-fixed-point}
\end{align}
that is, $ \lim_n \reachone{P_{f(n)}}$ is a fixed point of $\Gamma_P$.  Since $\reachone{P}$ is the least fixed point of $\Gamma_P$, we can conclude that $\reachone{P} \sqsubseteq  \lim_n \reachone{P_{f(n)}}$.  Hence,
\begin{align}
\reachone{P} 
& \sqsubseteq \lim_n \reachone{P_{f(n)}}
\label{equation:P-smaller-limit}\\
& = \lim_n \delta_{\tau_{\varepsilon_{f(n)}}}
\proofcomment{$\reachone{P_{f(n)}} = \delta_{\tau_{\varepsilon_{f(n)}}}$ since $P_{f(n)}$ is optimal}
\label{equation:gamma-leq-lim-delta}
\end{align}

Since $(\varepsilon_{f(n)})_n$ converges to 0, we can conclude from Proposition~\ref{proposition:tau-affine} that $(\tau_{\varepsilon_{f(n)}})_n$ converges to $\tau$.   Let $s$, $t \in S$.  To conclude that $\delta_{\_}(s, t)$ is discontinuous at $\tau$ it suffices to show that $\delta_\tau(s, t) < \lim_n \delta_{\tau_{\varepsilon_{f(n)}}}(s, t)$.  We distinguish the following two cases.
\begin{itemize}
\item 
Assume that the policy $P$ is not optimal for $(s, t)$.  Then, we have that $\delta_\tau(s, t) < \reachone{P}(s, t) \leq \lim_n \delta_{\tau_{\varepsilon_{f(n)}}}(s, t)$ by Theorem~\ref{proposition:minimal-coupling} and (\ref{equation:gamma-leq-lim-delta}).
\item 
Assume that the vertex policy $P$ is optimal for $(s, t)$.
Then, by Theorem~\ref{theorem:reach-s1}, we have that $(s, t)$ reaches $S^2_1$ with probability $\delta_\tau (s, t)$ in $\chain{P}$. It follows from (\ref{equation:reach-delta-discontinuous}) and Proposition~\ref{proposition:ccc} that $(s, t)$ can reach a closed communication class $C$ of $\chain{P}$ with $C \subseteq S^2_{0,\tau}$.
We consider a path from $(s, t)$ to $C$.  For each $(u, v)$ on that path, we prove that 
\begin{equation}
\label{equation:delta-ls-lim-gamma}
\reachone{P}(u, v) < \lim_n \reachone{P_{f(n)}}(u, v)
\end{equation}
by induction on the length of the path from $(u, v)$ to $C$.
\begin{itemize}
\item 
In the base case, $(u, v) \in C$.   Since the subsequence $(\reachone{P_{f(n)}}(u, v))_n$ converges, we can conclude from Proposition~\ref{proposition:discontinuity-gamma} that $\lim_n \reachone{P_{f(n)}}(u, v) > 0$.   Since $(u, v) \in C \subseteq S^2_{0,\tau}$, we can conclude from Theorem~\ref{theorem:distance-zero} that $\delta_\tau(u, v) = 0$.   Since $P$ is optimal for $(s, t)$ and $(u, v)$ is reachable from $(s, t)$, $P$ is also optimal for $(u, v)$ by Proposition~\ref{proposition:reachable-optimal}.  Therefore, $\reachone{P}(u, v) = \delta_\tau(u, v) = 0 < \lim_n \reachone{P_{f(n)}}(u, v)$.
\item 
In the inductive case, assume that $(u, v) (w, x) \ldots$ is a suffix of the path from $(s, t)$ to $C$.  Then $P(u, v)(w, x) > 0$.  By the induction hypothesis, $\reachone{P}(w, x) < \lim_n \reachone{P_{f(n)}}(w, x)$.  Therefore,
\begin{align*}
\reachone{P}(u, v)
& = \Gamma_P(\reachone{P})(u, v)\\
& = P(u, v) \cdot \reachone{P}\\
& = P(u, v)(w, x) \, \reachone{P}(w, x) + \sum_{(y, z) \not= (w, x)} P(u, v)(y, z) \, \reachone{P}(y, z)\\
& < P(u, v)(w, x) \, \lim_n \reachone{P_{f(n)}}(w, x) + \sum_{(y, z) \not= (w, x)} P(u, v)(y, z) \, \reachone{P}(y, z)\\
& \quad \proofcomment{$P(u, v)(w, x) > 0$ and $\reachone{P}(w, x) < \lim_n \reachone{P_{f(n)}}(w, x)$}\\
& \leq P(u, v)(w, x) \, \lim_n \reachone{P_{f(n)}}(w, x) + \sum_{(y, z) \not= (w, x)} P(u, v)(y, z) \, \lim_n \reachone{P_{f(n)}}(y, z)\\
& \quad \proofcomment{(\ref{equation:P-smaller-limit})}\\
& = \Gamma_P(\lim_n \reachone{P_{f(n)}})(u, v)\\
& = \lim_n \reachone{P_{f(n)}}(u, v)
\proofcomment{(\ref{equation:lim-gamma-fixed-point})}
\end{align*}
\end{itemize}
From the above we can conclude that
\begin{align*}
\delta_\tau(s, t)
& \leq \reachone{P}(s, t)
\proofcomment{Theorem~\ref{proposition:minimal-coupling}}\\
& < \lim_n \reachone{P_{f(n)}}(s, t)
\proofcomment{(\ref{equation:delta-ls-lim-gamma})}\\
& = \lim_n \delta_{\tau_{\varepsilon_{f(n)}}}(s, t) \proofcomment{$\reachone{P_{f(n)}}= \delta_{\tau_{\varepsilon_{f(n)}}}$ since $P_{f(n)}$ is optimal} \qedhere
\end{align*}
\end{itemize}
\end{proof}

\begin{proof}[Proof of \cref{theorem:main-continuity}]
Let $s$, $t \in S$.
\begin{enumerate}
    \item Follows from Proposition~\ref{proposition:phi-impossible}.
    \item Follows directly from Proposition~\ref{proposition:continuous} and Lemma~\ref{lemma:discontinuous}. \qedhere
\end{enumerate}
\end{proof}

\begin{theorem}
\label{theorem:robust-bisimilarity}
For all $s$, $t \in S$, $s \simeq t$ if and only if $s \sim t$ and $\delta_{\_}(s, t) : (S \to \Dreal(S)) \to [0, 1]$ is continuous at $\tau$.
\end{theorem}
\begin{proof}
Follows directly from Definition~\ref{definition:robust-probabilistic-bisimilarity} and \cref{theorem:main-continuity}. 
\end{proof}

For $\mu$, $\nu \in \Sreal(S)$ and an equivalence relation $R \subseteq S \times S$ such that for all $R$-equivalence classes $A$, $\mu(A) = \nu(A)$, we say that $\omega \in \Creal(\mu, \nu)$ is a \emph{maximal $R$-support coupling} if $\support(\omega) = (\support(\mu) \times \support(\nu)) \cap R$.  Moreover, we say that $P \in \mathcal{P}$ is a \emph{maximal $R$-support policy} if for all $(s, t) \in R \cap S^2_{0?}$, the coupling $P(s, t)$ is a maximal $R$-support coupling, that is, $\support(P(s, t)) = \mathrm{Post}((s, t)) \cap R$.

\begin{proposition}[{\cite[Proposition~19]{FKPB25}}]
\label{proposition:maximal-support-coupling}
For all $\mu$, $\nu \in \Sreal(X)$, given an equivalence relation $R \subseteq X \times X$ such that for all $R$-equivalence classes $A$, $\mu(A) = \nu(A)$, there exists a maximal $R$-support coupling $\omega \in \Creal(\mu, \nu)$.
\end{proposition}

\begin{proposition}[{\cite[Proposition~20]{FKPB25}}]
\label{proposition:maximal-support-policy}
For any bisimulation $R \subseteq S \times S$, a maximal $R$-support policy $P \in \mathcal{P}$ exists.
\end{proposition}

\begin{proof}[Proof of \cref{corollary:continuity-trivial}]
We prove the two implications.

Suppose that $\mathord{\sim} = \robust$. Let $s$, $t \in S$ and $P \in \mathcal{P}$ be an optimal policy for $(s, t)$.
Let $Q \in \mathcal{P}$ be an $S^2_\Delta$-closed maximal $\sim$-support policy.  We construct $P' \in \mathcal{P}$ as follows.
\[
P'(x, y) = \left \{
\begin{array}{ll}
Q(x, y) & \hspace{0.5cm} \mbox{if $x \sim y$}\\
P(x, y) & \hspace{0.5cm} \mbox{otherwise.}
\end{array}
\right .
\]
Note that, since $P'$ is also optimal for $(s, t)$, $(s, t)$ reaches $S^2_1$ with probability $\delta_\tau(s, t)$ in $\chain{P'}$.
It follows from the construction of $P'$ and \cref{proposition:closed-communication-class} that $(s, t)$ reaches $S^2_\Delta$ with probability $1 - \delta_\tau(s, t)$ in $\chain{P'}$.
Therefore, according to \cref{theorem:main-continuity}b, $\delta_{\_}(s, t) : (S \to \Dreal(S)) \to [0, 1]$ is continuous at $\tau$.

Suppose that $\mathord{\sim} \neq \robust$.  Since $\robust \subseteq \mathord{\sim}$, there exists $(s, t) \in \mathord{\sim} \setminus \robust$. By Theorem~\ref{theorem:robust-bisimilarity}, $\delta_{\_}(s, t) : (S \to \Dreal(S)) \to [0, 1]$ is discontinuous at $\tau$.
\end{proof}